\documentclass{article}
\usepackage{amssymb}
\usepackage{mathrsfs}
\usepackage[centertags]{amsmath}
\usepackage{amsthm}
\usepackage{graphics}
\usepackage{epsfig}

\hfuzz2pt \theoremstyle{plain}
\newtheorem{thm}{Theorem}[section]

\newtheorem{coro}[thm]{Corollary}
\newtheorem{lem}[thm]{Lemma}
\newtheorem{prop}[thm]{Proposition}
\theoremstyle{definition}

\newtheorem{problem}{Problem}[section]
\theoremstyle{remark}
\numberwithin{equation}{section}
\renewcommand{\theequation}{\thesection.\arabic{equation}}
\def\R{\mathbb{R}}

\DeclareMathOperator*{\var}{\mathbf{Var}}
\DeclareMathOperator*{\mini}{\mathrm{Minimize}}
\DeclareMathOperator*{\dd}{\mathrm{d\!}}

\DeclareMathOperator{\defi}{\stackrel{\mathrm{def}}{=\!\!=}}

\DeclareMathOperator*{\BE}{\mathbf{E}}
\input tcilatex
\begin{document}

\title{Continuous-Time Markowitz's Model with
Transaction Costs}
\date{}
\author{Min Dai\thanks{Department of Mathematics, National University of Singapore (NUS), Singapore. Min Dai is also an affiliated member of Risk Management Institute of NUS and is partially supported by Singapore MOE AcRF grant (No. R-146-000-096-112) and
NUS RMI grant (No. R-146-000-117-720/646).} \and
Zuo Quan Xu\thanks{Mathematical Institute and Nomura Centre for
Mathematical Finance, University of Oxford, 24--29 St Giles, Oxford
OX1 3LB.} \and Xun Yu Zhou\thanks{Nomura Centre for Mathematical
Finance and Oxford--Man Institute of Quantitative Finance,
University of Oxford, 24--29 St Giles, Oxford OX1 3LB, and
Department of Systems Engineering and Engineering Management, The
Chinese University of Hong Kong, Shatin, Hong Kong.
Email:$<$zhouxy@maths.ox.ac.uk$>$.}} \maketitle

\begin{abstract}
A continuous-time Markowitz's mean-variance portfolio selection
problem is studied in a market with one stock, one bond, and
proportional transaction costs. This is a singular stochastic
control problem, inherently in a finite time horizon. With a series
of transformations, the problem is turned into a so-called double
obstacle problem, a well studied problem in physics and partial
differential equation literature, featuring two time-varying free
boundaries. The two boundaries, which define the buy, sell, and
no-trade regions, are proved to be smooth in time. This in turn
characterizes the optimal strategy, via a Skorokhod problem, as one
that tries to keep a certain adjusted bond--stock position within
the no-trade region. Several features of the optimal strategy are
revealed that are remarkably different from its no-transaction-cost
counterpart.
%With the aid of the Lagrangian multiplier method and a
%partial differential equation approach, we obtain the necessary and
%sufficient condition for the existence of an optimal solution, and reveal
%three distinctive features of the Markowitz model in the presence of
%transaction costs.
It is shown that there exists a critical length in time, which is dependent on the
stock excess return as well as the transaction fees but {\it independent} of the
investment target and the stock volatility,
%is precisely given in terms of the market parameters,
so that an expected terminal return may not be achievable if the
planning horizon is shorter than that critical length (while in the
absence of transaction costs any expected return can be reached in
an arbitrary period of time). It is further demonstrated that anyone
following the optimal strategy should not buy the stock beyond the
point when the time to maturity is shorter than the aforementioned
critical length. Moreover, the investor would be less likely to buy
the stock and more likely to sell the stock when the maturity date
is getting closer. These features, while consistent with the widely
accepted investment wisdom, suggest that the planning horizon is an
integral part of the investment opportunities.
\end{abstract}


\section{Introduction{\ }}

Markowitz's (single-period) mean--variance (MV) portfolio selection model [Markowitz (1952)]
marked the start of the
modern quantitative finance theory.
Perverse enough, extensions to the dynamic
-- especially continuous-time -- setting in the asset allocation literature
have been dominated by the expected utility maximization (EUM) models, which take
a considerable departure from the MV model.
While the utility approach was theoretically justified by von Neumann and
Morgenstern (1947), in practice ``few if any investors know their utility functions;
nor do the functions which financial engineers and financial economists find analytically
convenient necessarily represent a particular investor's
attitude towards risk and return'' [Markowitz (2004)].
%Technically, it is wrong to think that the MV model is nothing else than
%a special case of the quadratic EUM model, since the
Meanwhile, there are technical difficulties in studying a dynamic MV
model, primarily that of the incompatibility with the dynamic
programming principle owing to the variance term involved. In other
words, an optimal trading strategy generated initially may no longer
be optimal half way through. However, this time inconsistency only
means that the dynamic programming -- which is a {\it tool} for
solving dynamic optimization problems -- is not applicable (at least
not directly applicable). The dynamic mean--variance problem, while
lacking the time consistency, is economically sound since its
essence is to find a Pareto optimal strategy that strikes a balance
between the return and risk at the end of the investment planning
horizon. Indeed, there are many other problems that are inherently
time inconsistent. For example, a dynamic behavioral portfolio
selection problem is time inconsistent due to the distortions in
probabilities [Jin and Zhou (2008)].  In a broader economical
perspective, Kydland and  Prescott (1977) pointed out as early as in
the 1970s that time consistent policies may be suboptimal for some
economic problems. Mean--variance is exactly one of these problems.
Recently time-inconsistent control problems have attracted
considerable research interest and effort; see, e.g., Basak and
Chabakauri (2008), Bj\"ork (2008), and Ekeland and Lazrak (2007).

Richardson (1989) is probably the earliest paper that studies a
faithful extension of the MV model to the continuous-time setting
(albeit in the context of a single stock with a constant risk-free
rate), followed by Bajeux-Besnainou and Portait (1998). Li and Ng
(2000) developed an embedding technique to cope with the
time inconsistency for a discrete-time MV
model, which was extended by Zhou and Li (2000), along with a
stochastic linear--quadratic control approach, to the
continuous-time case. Further extensions and improvements are
carried out in, among many others, Lim and Zhou (2002), Lim
(2004), Bielecki {\it et al.} (2005), and Xia (2005).

All the existing works on continuous-time MV models have assumed
that there is no transaction cost, leading to results that are
analytically elegant, and sometimes truly surprising [for example,
it is shown in Li and Zhou (2006) that any efficient strategy
realizes its goal -- no matter how high it is -- with a
probability of at least 80\%]. However, elegant they may be,
certain investment behaviors derived from the results simply
contradict the conventional wisdom, which in turn hints that the
models may not have been properly formulated. For instance, the
results dictate that an optimal strategy must trade all the time;
moreover, there must be risky exposures at any time [see Chiu and
Zhou (2007)]. These are certainly not consistent with the common
investment advice. Indeed, the assumption that there is no
transaction cost is flawed, which misleadingly allows an investor
to continuously trade without any penalty.

Portfolio selection subject to transaction costs has been studied
extensively, albeit in the realm of utility maximization.
Mathematically such a problem is a singular stochastic control
problem. Two different types of models must be distinguished: one in
an infinite planning horizon and the other in a finite horizon. See
Magill and Constantinides (1976), Davis and Norman (1990), and
Shreve and Soner (1994) for the former, and Davis, Panas and
Zariphopoulou (1993), Civitanic and Karatzas (1996), and Gennotte
and Jung (1994) for the latter. Technically, the latter is
substantially more difficult than the former, since in the finite
horizon case there is an additional time variable in the related
Hamilton-Jacobi-Bellman (HJB) equation or variational inequality
(VI). This is why the research on finite-horizon problems had been
predominantly on qualitative and numerical solutions until Liu and
Loewenstein (2002) devised an analytical approach based on an
approximation of the finite horizon by a sequence of Erlang
distributed random horizons. Dai and Yi (2009) subsequently employed
a different analytical approach -- a PDE one -- to study the same
problem.

This paper aims to analytically solve the MV model with transaction costs. Note that
such a problem is inherently one in a finite time horizon, because the very nature of
the Markowitz problem is about striking a balance between the risk and return of the wealth at a finite, {\it terminal} time.
Compared with its EUM counterpart, there is a feasibility issue that must be addressed before
an optimal solution is sought. Precisely speaking, the MV model is to minimize the variance
of the terminal wealth subject to the constraint that an investment target -- certain expected net
terminal wealth -- is achieved. The feasibility is about whether such a target
is achievable by at least one admissible investment strategy. For a Black--Scholes market
without transaction costs, it has been shown [Lim and Zhou (2002)]
that {\it any} target can be reached in an
arbitrary length of time (so long as the risk involved is not a concern, that is).
For a more complicated model with random investment opportunities and no-bankruptcy
constraint, the feasibility is painstakingly investigated in Bielecki {\it et al} (2005).
In this paper we show that the length of the planning horizon
is a determinant of this issue.
In fact, there exists a critical length of time, which is dependent on the
stock excess return as well as the transaction fees but {\it independent} of the
investment target and stock volatility,
%is precisely given in terms of the market parameters,
so that a sufficiently high target is not achievable
if the planning horizon is shorter than that critical length. This certainly makes
good sense intuitively.

To obtain an optimal strategy, technically we follow the idea of Dai
and Yi (2009) of eventually turning the associated VI into a
double-obstacle problem, a problem that has been well studied in
physics and PDE theory. That said, there are indeed intriguing
subtleties when actually carrying it out. In particular, this paper
is the first to prove (to the best of the authors' knowledge) that
the two free boundaries that define the buy, sell and no-trade
regions are smooth. This smoothness is critical in deriving the
optimal strategy via a Skorokhod problem.\footnote{This smoothness
also plays a crucial role in studying the finite horizon optimal
investment and consumption with transaction costs under utility
framework, see Section 3 of Dai et al. (2009) where an integral over
one free boundary is used.} The optimal strategy is rather simple in
implementation; it is to keep a certain adjusted bond--stock
position within the no-trade region. Several features of the optimal
strategy are revealed that are remarkably different from its
no-transaction-cost counterpart. Among them it is notable that one
should no longer buy stock beyond the point when the time to
maturity is shorter than the aforementioned critical length
associated with the feasibility. Moreover, one is less likely to buy
the stock and more likely to sell the stock when the maturity date
is getting closer. These are consistent with the widely accepted
financial advice, and suggest that the planning horizon should be
regarded as a part of the investment opportunity set when it comes
to continuous time portfolio selection.

The remainder of the paper is organized as follows. The model under consideration is formulated
in section 2, and
the feasibility issue is addressed in section 3. The optimal strategy is derived
in sections 4--6 via several steps, including Lagrange relaxation, transformation of the
HJB equation
to a double obstacle problem, and the Skorokhod problem. Finally, the
paper is concluded with remarks in section 7. Some technical proofs are relegated to an appendix.

\section{Problem Formulation}

We consider a continuous-time market where
there are only two investment instruments: a bond and a stock with
price dynamics given respectively by
\begin{align*}
\dd R(t)& =rR(t)\dd t, \\
\dd S(t)& =\alpha S(t)\dd t+\sigma S(t)\dd B(t).
\end{align*}
Here $r>0$, $\alpha >r$ and $\sigma >0$ are constants, and the process $%
\{B(t)\}_{t\in [0,T]}$ is a standard one-dimensional Brownian motion on a
filtered probability space $(\Omega ,\mathcal{F},\{\mathcal{F}_t\}_{t\in
[0,T]},\mathbf{P})$ with $B(0)=0$ almost surely. We assume that %$\mathcal{F}=%\mathcal{F}_{T}$,
the filtration $\{\mathcal{F}_t\}_{t\in [0,T]}$ is generated by the
Brownian motion, is right continuous, and each $\mathcal{F}_t$
contains all the $\mathbf{P}$-null sets of $\mathcal{F}$. We denote
by  ${L}_{\mathcal{F}}^2$ the set of square integrable
$\{\mathcal{F}_t\}_{t\in [0,T]}$-adapted processes, 
\begin{align*}
{L}^{2}_{\mathcal{F}}\defi\left\{ X\left| \;
\parbox{80mm}{ The process $X=\{X(t)\}_{t\in [0,T]}$ is an $\{\mathcal{F}_t\}_{t\in[0,T]}$-adapted process such that $\int_{0}^{T}\BE\left[X^{2}(t)\right]\dd t<\infty$ }\right. \right\},
\end{align*}
and by
${L}_{\mathcal{F}_T}^2$ the set of square integrable
$\mathcal{F}_T$-measurable random variables,
\begin{align*}
{L}^{2}_{\mathcal{F}_{T}}\defi\left\{ X\left| \; \parbox{100mm}{ X is an $\mathcal{F}_T$-measurable random variable such that $\BE\left[X^{2}\right]<\infty$ }\right. \right\}.
\end{align*}
\par
There is a self-financing investor with a finite investment horizon $[0,T]$
who invests $X(t)$ dollars in the bond
and $Y(t)$ dollars in the stock at time $t$.
Any stock transaction incurs a \textit{proportional transaction fee}, with
$\lambda \in [0,+\infty )$ and $\mu \in [0,1)$
being the proportions paid
when buying and selling the stock, respectively.
Throughout this paper, we
assume that $\lambda +\mu >0$, which means transaction costs must be
involved.
The bond--stock value process, starting from $\left( x,y\right)$ at $t=0$,
evolves according to the equations:
\begin{eqnarray}
X^{x,M,N}(t)&=&x+r\int_0^t\!\!X^{x,M,N}(s)\dd s-(1+\lambda )M(t)+(1-\mu
)N(t),  \label{xyprocess} \\
Y^{y,M,N}(t)&=&y+\alpha \int_0^t\!\!Y^{y,M,N}(s)\dd s+\sigma
\int_0^t\!\!Y^{y,M,N}(s)\dd B(s)+M(t)-N(t),  \label{xyprocess2}
\end{eqnarray}
where $M(t)$ and $N(t)$ denote respectively the cumulative stock purchase and sell up to time $t$.
Sometimes we simply use $X$, $Y$ or $X^{M,N}$, $Y^{M,N}$ instead of $X^{x,M,N}$, $Y^{y,M,N}$
if there is no ambiguity.

The \textit{admissible strategy set} $\mathcal{A}$ of the investor is
defined as follows:
\begin{align*}
\mathcal{A}\defi\left\{ (M,N)\left| \;%
\parbox{100mm}{ The processes
$M=\{M(t)\}_{t\in [0,T]}$ and $N=\{N(t)\}_{t\in [0,T]}$ are
$\{\mathcal{F}_t\}_{t\in[0,T]}$-adapted, RCLL, nonnegative and
nondecreasing, and the processes $X^{x,M,N}$ and $Y^{y,M,N}$ are
both in $L^2_{\mathcal{F}}$, for any $(x,y)\in\R^2$ }\right. \right\} .
\end{align*}
$(M,N)$ is called \textit{an admissible strategy} if $(M,N)\in \mathcal{A}$.
Correspondingly, $(X^{x,M,N},Y^{y,M,N})$ is called an \textit{admissible (bond--stock)
process} if $(x,y)\in \mathbb{R}^2$ and $(M,N)\in \mathcal{A}$.  
\par
For an admissible process $(X^{x,M,N},Y^{y,M,N})$, we define the investor's
\textit{net wealth process} by
\begin{align*}
W^{X,Y}(t)\defi X(t)+(1-\mu )Y(t)^{+}-(1+\lambda )Y(t)^{-},\quad t\in [0,T].
\end{align*}
Namely, $W^{X,Y}(t)$ is the net worth of the investor's portfolio at $t$ after the
transaction cost is deducted.
The investor's \textit{attainable net wealth set} at the maturity time $T$
is defined as
\begin{align*}
\mathcal{W}_0^{x,y}\defi \left\{ W^{X,Y}(T)\left| \;%
\parbox{65mm}{$W^{X,Y}(T)$ is the net wealth at $T$ of an admissible process $(X,Y)$ with
$X(0-)=x$, $Y(0-)=y$.}\right. \right\}. % ,\quad s\in [0,T).
\end{align*}

% In the spirit of the MV portfolio selection, the investor has
% a prescribed expected net wealth $z>0$ at the terminal time $T$, while
% his objective is to minimize the variance of the terminal net wealth. This
% can be formulated as follow:

In the spirit of the original Markowitz's MV portfolio theory, an {\it efficient strategy} is a trading strategy
for which there does not exist
another strategy that has higher mean and no higher variance,
and/or has less variance and no less mean at the terminal time $T$.
In other words, an efficient strategy is one that is {\it Pareto
optimal}. Clearly, there could be many efficient strategies, and  the terminal means and variances corresponding to all the efficient strategies form
an {\it efficient frontier}. The positioning on the efficient frontier of a specific investor is dictated by his/her risk preference.

It is now well known that the efficient frontier can be obtained
from solving the following {\it variance minimizing} problem:

\begin{problem}
%[\textbf{Mean-Variance Problem}]
\label{pmeanvariance0}
\begin{eqnarray*}
\begin{array}{lll}
& \mini & \var(W), \\
& \mathrm{subject\;to} & \BE[W]=z,\;\; W\in \mathcal{W}_0^{x,y}.
\end{array}
\end{eqnarray*}
\end{problem}
Here $z$ is a parameter satisfying
\begin{align*}
z> e^{rT}x+(1-\mu )e^{rT}y^{+}-(1+\lambda )e^{rT}y^{-},  \label{daim1}
\end{align*}
which means that the target expected  terminal wealth is higher than
that of the simple ``all-bond'' strategy (i.e.
initially liquidating the stock investment and
putting all the money in the bond
account). The optimal solutions to the above problem with varying values of $z$ will trace out the efficient frontier we are looking for.
For this reason, although Problem 2.1 is indeed an auxilary mathematical problem introduced to help solve the original mean--variance problem, it is sometimes
(as in this paper) itself called the mean--variance problem.

It is immediate to see that
Problem 2.1 is equivalent to the following problem.

\begin{problem}
\label{pmeanvariance}
\begin{eqnarray*}
\begin{array}{lll}
& \mini & \BE[W^2], \\
& \mathrm{subject\;to} & \BE[W]=z,\;\;W\in \mathcal{W}_0^{x,y}.
\end{array}
\end{eqnarray*}
\end{problem}

\section{Feasibility}

In contrast with the EUM problem, the MV model,
Problem \ref{pmeanvariance}, has an inherent constraint $\BE[W]=z$. Is there always
an admissible strategy to meet this constraint no matter how aggressive the target $z$ is?
This is the so-called feasibility issue. The issue is important and unique to the
MV problem, and will be
addressed fully in this
section. To begin with, we introduce two lemmas.

\begin{lem}
\label{Wincreasing} If $W_1\in \mathcal{W}_0^{x,y}$, $%
W_2\in L_{\mathcal{F}_T}^2$ and $W_2\leqslant W_1$, then $W_2\in \mathcal{W}%
_0^{x,y}$.
\end{lem}

\begin{proof}
By the definition of $\mathcal{W}_0^{x,y}$, there exists
$(M,N)\in\mathcal{A}$ such that $X^{M,N}(0-)=x$, $Y^{M,N}(0-)=y$ and
$W^{X^{M,N},Y^{M,N}}(T)=W_1$. We define
\begin{align*}
\overline{M}(t)=
\begin{cases}
M(t), & \textrm{if } t<T,\\
M(T)+\dfrac{W_1-W_2}{\lambda+\mu},& \textrm{if } t=T,
\end{cases}\qquad
\overline{N}(t)=
\begin{cases}
N(t), & \textrm{if } t<T,\\
N(T)+\dfrac{W_1-W_2}{\lambda+\mu},& \textrm{if } t=T.
\end{cases}
\end{align*}
Then $(\overline{M},\overline{N})\in\mathcal{A}$ and
\begin{align*}
X^{\overline{M},\overline{N}}(t)=
\begin{cases}
X^{M,N}(t), & \textrm{if } t<T,\\
X^{M,N}(T)-W_1+W_2,& \textrm{if } t=T,
\end{cases}\qquad
Y^{\overline{M},\overline{N}}(t)=Y^{M,N}(t),\quad t\in[0,T].
\end{align*}
Therefore $W_2=W^{X^{\overline{M},\overline{N}},Y^{\overline{M},\overline{N}}}(T)
\in\mathcal{W}_0^{x,y}$.
\end{proof}

The proof above is very intuitive. If a higher terminal wealth is
achievable by an admissible strategy, then so is a lower one, by
simply ``wasting money'', i.e., buying and selling the same amount
of the stock at $T$, thanks to the presence of the transaction
costs. This is {\it not} necessarily true when there is no
transaction cost.

\begin{lem}
\label{Wconvexset} For any $(x,y)\in \mathbb{R}^2$, we have

(1) the set $\mathcal{W}_0^{x,y}$ is convex;

(2) %%$\mathcal{W}_t^{\cdot ,\cdot }$ is convex in $\mathbb{R}^2$;
If $(x_i,y_i)\in\R^2$, $W_i\in\mathcal{W}_0^{x_i, y_i}$, $i=1$, $2$, then
$W_1+W_2\in\mathcal{W}_0^{ x_1+x_2, y_1+y_2}$;

(3) if $x_1\leqslant x_2$ and $y_1\leqslant y_2$, then $\mathcal{W}%
_0^{x_1,y_1}\subseteq \mathcal{W}_0^{x_2,y_2}$;

(4) $\mathcal{W}_0^{x-(1+\lambda )\rho ,y+\rho }\subseteq \mathcal{W}_0^{x,y}
$ and $\mathcal{W}_0^{x+(1-\mu )\rho ,y-\rho }\subseteq \mathcal{W}_0^{x,y}$
for any $\rho >0$;

(5) $\mathcal{W}_0^{\rho x,\rho y}=\rho \mathcal{W}_0^{x,y}$ for any $\rho >0
$;

(6) if $x+(1-\mu )y^{+}-(1+\lambda )y^{-}\geqslant 0$, then $0\in \mathcal{W}%
_0^{x,y}$.
\end{lem}

\begin{proof}
%We only prove the case of $t=0$.
(1) For any $W_1$, $W_2\in \mathcal{W}_0^{x,y}$,
assume that $W_i=W^{X_i,Y_i}(T)$,
where $(X_i,Y_i)=(X^{x,M_i,N_i},Y^{y,M_i,N_i})$, $(M_i,N_i)\in\mathcal{A}$, $i=1,2$.
For any $k\in(0,1)$, let $M=kM_1+(1-k)M_2$, $N=kN_1+(1-k)N_2$.
Then $(M,N)\in\mathcal{A}$, and
\begin{align*}
X^{x,M,N}&=kX_1+(1-k)X_2,\\
Y^{y,M,N}&=kY_1+(1-k)Y_2.
\end{align*}
Thus
\begin{equation*}
\begin{split}
W^{X,Y}(T)&=X(T)+(1-\mu)Y(T)^+-(1+\lambda)Y(T)^-
=X(T)+(1-\mu)Y(T)-(\mu+\lambda)Y(T)^-\\
&\geqslant kX_1(T)+(1-k)X_2(T)+(1-\mu)(kY_1(T)+(1-k)Y_2(T))\\
&\quad-(\mu+\lambda)(kY_1(T)^-+(1-k)Y_2(T)^-)\\
&=k(X_1(T)+(1-\mu)Y_1(T)^+-(1+\lambda)Y_1(T)^-)\\
&\quad+(1-k)(X_2(T)+(1-\mu)Y_2(T)^+-(1+\lambda)Y_2(T)^-)\\
&=kW_1+(1-k)W_2.
\end{split}
\end{equation*}
Since $W^{X,Y}(T)\in\mathcal{W}_0^{x,y}$, $kW_1+(1-k)W_2\in
L^2_{\mathcal{F}_T}$, it follows from Lemma \ref{Wincreasing} that
$kW_1+(1-k)W_2\in\mathcal{W}_0^{x,y}$.
\par
(2) This can be proved by the same argument as above.
\par
(3) If $x_1\leqslant x_2$ and $y_1\leqslant y_2$, then for any $(M,N)\in\mathcal{A}$, we have
$X^{x_1,M,N}\leqslant X^{x_2,M,N}$, $Y^{y_1,M,N}\leqslant Y^{y_2,M,N}$.
So $W^{X^{x_1,M,N}, Y^{y_1,M,N}}(T)\leqslant
W^{X^{x_2,M,N}, Y^{y_2,M,N}}(T)\in \mathcal{W}_0^{x_2,y_2}$.
Therefore by Lemma \ref{Wincreasing},
we have $W^{X^{x_1,M,N}, Y^{y_1,M,N}}(T)\in \mathcal{W}_0^{x_2,y_2}$.
This shows $\mathcal{W}_0^{x_1,y_1}\subseteq \mathcal{W}_0^{x_2,y_2}$.
\par
(4) For any $\rho>0$, $(M,N)\in\mathcal{A}$, we define
\begin{align*}
\overline{M}(s)&=
M(s)+\rho, \quad\forall s\in[0,T].
 \end{align*}
Then $(X^{x-(1+\lambda)\rho,M,N},Y^{y+\rho,M,N})=(X^{x,\overline{M},N},Y^{y,\overline{M},N})$.
So $W^{X^{x-(1+\lambda)\rho,M,N},Y^{y+\rho,M,N}}=W^{X^{x,\overline{M},N},Y^{y,\overline{M},N}}
\in \mathcal{W}_0^{x,y}$. Hence
$\mathcal{W}_0^{x-(1+\lambda)\rho,y+\rho}\subseteq \mathcal{W}_0^{x,y}$.
Similarly, we can prove that
$\mathcal{W}_0^{x+(1-\mu)\rho,y-\rho}\subseteq \mathcal{W}_0^{x,y}$.
\par
(5) Noting that
$W^{X^{\rho x,\rho M,\rho N},Y^{\rho y,\rho M,\rho N}}=\rho W^{X^{x,M,N},Y^{y,M,N}}$ $\forall \rho>0$,
we have immediately $\mathcal{W}_0^{\rho x,\rho y}=\rho\mathcal{W}_0^{x,y}$.
\par
(6) If $y\geqslant 0$ and $x+(1-\mu)y\geqslant 0$, then obviously
$0\in\mathcal{W}_0^{x+(1-\mu)y,0}$ by Lemma \ref{Wincreasing}.
Noting (4) proved above, we conclude that $0\in\mathcal{W}_0^{x,y}$.
Similarly we can prove the case of $y<0$ and $x+(1+\lambda)y\geqslant 0$.
\end{proof}
\hspace{1.0in}

Denote
\begin{align}
\hat z\defi\sup \left\{ \BE[W]\;\big|\;W\in \mathcal{W}_0^{x,y}\right\} .
\label{hatzdefinition}
\end{align}
In view of Lemma \ref{Wincreasing}, Problem \ref{pmeanvariance} is feasible when
\begin{align}
z\in \mathcal{D}\defi\left( e^{rT}x+(1-\mu )e^{rT}y^{+}-(1+\lambda
)e^{rT}y^{-},\;\hat z\right) .  \label{mathcalDdefinition}
\end{align}
It is clear that Problem \ref{pmeanvariance} is not feasible when $z>\hat z$%
. So, it remains to investigate whether Problem \ref{pmeanvariance} admits
a feasible solution when $z=\hat z$.

It is well known that in the absence of transaction costs (i.e. $\lambda=\mu=0$),
we have $\hat{z}=+\infty$ and thus Problem \ref{pmeanvariance}
is always feasible for any $z\geqslant e^{rT}x+(1-\mu)e^{rT}y^{+}-(1+\lambda )e^{rT}y^{-}$
[see Lim and Zhou (2002)]. In other words, no matter
how \textit{small} the investor's initial wealth is, the investor can always arrive at
an \textit{arbitrarily} large expected return in a \textit{split second,} by
taking a huge leverage on the stock. The
following theorem indicates that things become very different when the transaction costs
get involved.

\begin{thm}
\label{fesiblility}Assume $T\leqslant T^*\defi\frac 1{\alpha -r}\ln \left( \frac{%
1+\lambda }{1-\mu }\right) $. Then
\[
\hat z=%
\begin{cases} e^{rT}x+(1-\mu)e^{\alpha T}y, &\textrm{ if } y> 0,\\
%e^{rT}x+(1-\mu)e^{r T}y, &\textrm{ if }y\leqslant 0. \end{cases}
e^{rT}x+(1+\lambda)e^{r T}y, &\textrm{ if }y\leqslant 0. \end{cases}
\]
Moreover, if $y>0$ and $z=\hat z$, then Problem \ref{pmeanvariance0} admits
a unique feasible (thus optimal) solution and the optimal strategy is $%
(M,N)\equiv (0,0)$. If $y\leqslant 0$, then ${\mathcal{D}}=\emptyset$.
% and Problem \ref{pmeanvariance0} is infeasible.
\end{thm}

%%%%%%%%%%%%%%%%%%%% proof %%%%%%%%%%%%%%%%%%%%%%%%%%%%
%\iffalse
%%%%%%%%%%%%%%%%%%%%%%%%%%%%%%%%%%%%%%%%%%%%%%%%%%%%%%%

\begin{proof}
For any $(M,N)\in\mathcal{A}$, due to (\ref{xyprocess}),
(\ref{xyprocess2}) and It\^{o}'s formula, we have
\begin{equation*}
\begin{split}
%&\hspace{15pt}
X(T)+(1-\mu)Y(T) &=e^{rT}x+(1-\mu)e^{\alpha T}y
+\int_0^T\!\! e^{r(T-t)}\left(\dd X(t)-rX(t)\dd t\right)\\
&\quad+\int_0^T\!\!(1-\mu) e^{\alpha(T-t)}\left(\dd Y(t)-\alpha Y(t)\dd t\right)\\
&=e^{rT}x+(1-\mu)e^{\alpha T}y
+\int_0^T\!\!\left((1-\mu) e^{\alpha(T-t)}-(1+\lambda)e^{r(T-t)}\right)\dd M^c(t)\\
&\quad+\int_0^T\!\!(1-\mu)\left(e^{r(T-t)}-e^{\alpha(T-t)}\right)\dd
N^c(t)
+\int_0^T\!\!(1-\mu) e^{\alpha(T-t)}\sigma Y(t) \dd B(t)\\
&\quad+\sum\limits_{0\leqslant t\leqslant T}
\left((1-\mu) e^{\alpha(T-t)}-(1+\lambda)e^{r(T-t)}\right)(M(t)-M(t-))\\
&\quad+\sum\limits_{0\leqslant t\leqslant T}(1-\mu)
\left(e^{r(T-t)}-e^{\alpha(T-t)}\right)(N(t)-N(t-))\\
&\leqslant e^{rT}x+(1-\mu)e^{\alpha T}y +\int_0^T\!\!(1-\mu)
e^{\alpha(T-t)}\sigma Y(t) \dd B(t),
\end{split}
\end{equation*}
where we have noted $ T
\leqslant\frac{1}{\alpha-r}\ln\left(\frac{1+\lambda}{1-\mu}\right).
$ So
\begin{align*}
\BE [X(T)+(1-\mu)Y(T)]\leqslant e^{rT}x+(1-\mu)e^{\alpha T}y.
\end{align*}
It follows
\begin{equation*}
\begin{split}
\BE\left[W^{X,Y}(T)\right]&=\BE[X(T)+(1-\mu)Y(T)^+-(1+\lambda)Y(T)^-]\\
&\leqslant \BE[X(T)+(1-\mu)Y(T)] \leqslant  e^{rT}x+
(1-\mu)e^{\alpha T}y.
\end{split}
\end{equation*}
When $y\geqslant 0$, we have $\BE[W^{X,Y}(T)]=e^{rT}x+
(1-\mu)e^{\alpha T}y$ if and only if $(M,N)\equiv(0,0)$. Thus, if
$y\geqslant 0$, then $\hat{z}=e^{rT}x+ (1-\mu)e^{\alpha T}y$, and
$(0,0)$ is the unique feasible strategy when $z=\hat{z}$.
\par
Now we turn to the case of $y<0$. Define
$\tau=\inf\{t\in[0,T):Y(t)>0\}\wedge T$. Then $\tau$ is a stopping
time [cf. Theorem 2.33, Klebaner (2004)]. On the set
$\{Y(\tau)\geqslant 0\}$, by the same argument as above, we have
\begin{equation*}
\begin{split}
W^{X,Y}(T) &\leqslant e^{r(T-\tau)}X(\tau)+(1-\mu)e^{\alpha(T-\tau)}Y(\tau)\leqslant
e^{r(T-\tau)}X(\tau)+(1+\lambda)e^{r(T-\tau)}Y(\tau).
 \end{split}
\end{equation*}
%\begin{equation*}
%\begin{split}
%\BE\left[W^{X,Y}\big|Y(\tau)\geqslant 0\right] &\leqslant
%e^{r(T-\tau)}X(\tau)+(1-\mu)e^{\alpha(T-\tau)}Y(\tau)
%\leqslant e^{r(T-\tau)}X(\tau)+(1+\lambda)e^{r(T-\tau)}Y(\tau)\\
%&=e^{r(T-\tau)}X(\tau-)+(1+\lambda)e^{r(T-\tau)}Y(\tau-)+e^{r(T-\tau)}[-(1+\lambda)(M(\tau)-M(\tau-))\\
%&\quad+(1-\mu)(N(\tau)-N(\tau-))]
%+(1+\lambda)e^{r(T-\tau)}[(M(\tau)-M(\tau-))-(N(\tau)-N(\tau-))]\\
%&=e^{r(T-\tau)}X(\tau-)+(1+\lambda)e^{r(T-\tau)}Y(\tau-)
%-(\lambda+\mu)e^{r(T-\tau)}(N(\tau)-N(\tau-))\\
%&\leqslant e^{r(T-\tau)}X(\tau-)+(1+\lambda)e^{r(T-\tau)}Y(\tau-).
%\end{split}
%\end{equation*}
On the set $\{Y(\tau)<0\}$, we have $\tau=T$,
\begin{equation*}
\begin{split}
W^{X,Y}(T) &=X(T)+(1+\lambda) Y(T)=e^{r(T-\tau)}X(\tau)+(1+\lambda)e^{r(T-\tau)}Y(\tau).
\end{split}
\end{equation*}
%Then we obtain
%\begin{align*}
%\BE[W^{X,Y}|\mathcal{F}_{\tau}]&\leqslant
%e^{r(T-\tau)}X(\tau-)+(1+\lambda)e^{r(T-\tau)}Y(\tau-).
%\end{align*}
On the other hand, noting that $Y(t)\leqslant 0$, $t\in[0,\tau)$,
we have
\begin{equation*}
\begin{split}
&\hspace{15pt} e^{r(T-\tau)}X(\tau)+(1+\lambda)e^{r(T-\tau)}Y(\tau)\\
&=e^{rT}x+(1+\lambda)e^{rT}y-\int_0^{\tau}\!\!(\lambda+\mu)e^{r(T-t)}\dd
N^c(t) +\int_0^{\tau}\!\!(\alpha-r)e^{r(T-t)}Y(t)\dd t \\ &\quad
+\int_0^{\tau}\!\!e^{r(T-t)}(1-\mu)\sigma Y(t) \dd B(t)
-\sum\limits_{0\leqslant t\leqslant \tau}(\lambda+\mu)e^{r(T-t)}(N(t)-N(t-))\\
&\leqslant
e^{rT}x+(1+\lambda)e^{rT}y+\int_0^{\tau}\!\!e^{r(T-t)}(1-\mu)\sigma
Y(t) \dd B(t).
\end{split}
\end{equation*}
%\begin{equation*}
%\begin{split}
%&\hspace{15pt} e^{r(T-\tau)}X(\tau-)+(1+\lambda)e^{r(T-\tau)}Y(\tau-)\\
%&=e^{rT}x+(1+\lambda)e^{rT}y+\int_0^{\tau}\!\! e^{r(T-t)}(\dd
%X(t)-rx(t)\dd t)
%+\int_0^{\tau}\!\!(1+\lambda) e^{r(T-t)}(\dd Y(t)-r Y(t)\dd t)\\
%&=e^{rT}x+(1+\lambda)e^{rT}y-\int_0^{\tau}\!\!(\lambda+\mu)e^{r(T-t)}\dd
%N^c(t) +\int_0^{\tau}\!\!(\alpha-r)e^{r(T-t)}Y(t)\dd t \\ &\quad
%+\int_0^{\tau}\!\!e^{r(T-t)}(1-\mu)\sigma Y(t) \dd B(t)
%-\sum\limits_{0\leqslant t<\tau}(\lambda+\mu)e^{r(T-t)}(N(t)-N(t-))\\
%&\leqslant
%e^{rT}x+(1+\lambda)e^{rT}y+\int_0^{\tau}\!\!e^{r(T-t)}(1-\mu)\sigma
%Y(t) \dd B(t).
%\end{split}
%\end{equation*}
It follows
\begin{align*}
\BE[e^{r(T-\tau)}X(\tau)+(1+\lambda)e^{r(T-\tau)}Y(\tau)]\leqslant
e^{rT}x+(1+\lambda)e^{rT}y.
\end{align*}
Therefore,
\begin{equation*}
\begin{split}
\BE[W^{X,Y}(T)]& \leqslant
\BE[e^{r(T-\tau)}X(\tau)+(1+\lambda)e^{r(T-\tau)}Y(\tau)]
\leqslant e^{rT}x+(1+\lambda)e^{rT}y.
\end{split}
\end{equation*}
This indicates that $\BE[W^{X,Y}(T)]=e^{rT}x+ (1+\lambda)e^{r T}y$ if and
only if the investor puts all of his wealth in the bond at time 0.
Thus $\hat{z}=e^{rT}x+ (1+\lambda)e^{r T}y$ if $y<0$.
\end{proof}

%%%%%%%%%%%%%%%%%%%% proof %%%%%%%%%%%%%%%%%%%%%%%%%%%%
%\fi
%%%%%%%%%%%%%%%%%%%%%%%%%%%%%%%%%%%%%%%%%%%%%%%%%%%%%%%

This result demonstrates the importance of the length of the
investment planning horizon, $T$, by examining the situation when
$T$ is not long enough. In this ``short horizon'' case, if the
investor starts with a short position in stock, then the only
sensible strategy is the all-bond one, since any {\it other}
strategy will just be worse off in {\it both} mean and variance.
On the other hand, if one starts with a long stock position, then
the highest expected terminal net wealth (without considering the
variance) is achieved by the ``stay-put'' strategy, one that does
not switch at all between bond and stock from the very beginning.
Therefore, any efficient strategy is between the two extreme
strategies, those of ``all-bond'' and ``stay-put'', according to
an individual investor's risk preference.

More significantly, Theorem \ref{fesiblility} specifies explicitly this critical length of
horizon, $T^*=\frac 1{\alpha -r}\ln \left( \frac{%
1+\lambda }{1-\mu }\right)$. It is intriguing that $T^*$ depends
only on the excess return, $\alpha -r$, and the transaction fees
$\lambda,\mu$, \textit{not} on the individual target $z$ or the
stock volatility $\sigma$. Later we will show that, indeed,  $\hat
z=+\infty $ when $T>T^*$ in Corollary \ref{hatz=infty}. Therefore
$T^*$ is such a critical value in time that divides between ``global
feasibility'' and ``limited feasibility'' of the underlying MV
portfolio selection problem. It signifies the opportunity that a
longer time horizon would provide in achieving a higher potential
gain. In this sense, the length of the planning horizon should be
really included in the set of the investment opportunities, as
opposed to the hitherto widely accepted notion that the investment
opportunity set consists of only the probabilistic characteristics
of the returns. Moreover, it follows from the expression of $T^*$
that the less transaction cost and/or the higher excess return of
the stock the shorter time it requires to attain the global
feasibility. These, of course, all make perfect sense economically.

In the remaining part of this paper, we only consider the case when
${\mathcal{D}}\neq\emptyset$ and $z\in {\mathcal{D}}$.

\section{Unconstrained Problem and Double-Obstacle Problem}

\subsection{Lagrangian Relaxation and HJB equation}

By virtue of Lemma \ref{Wconvexset}, Problem \ref{pmeanvariance} is a convex
constrained optimization problem. We shall utilize the well-known Lagrange
multiplier method to remove the constraint.

Let us introduce the following unconstrained problem.

\begin{problem}[\textbf{Unconstrained Problem}]
\label{pmeanvariancewithoutconstrain0}
\begin{eqnarray*}
\begin{array}{lll}
& \mini & \BE[W^2]-2\ell (\BE[W]-z) \\
& \mathrm{subject\;to\;} & W\in \mathcal{W}_0^{x,y},
\end{array}
\end{eqnarray*}
\end{problem}
or equivalently,

\begin{problem}
\label{pmeanvariancewithoutconstrain}
\begin{eqnarray*}
\begin{array}{lll}
& \mini & \BE[(W-\ell)^2] \\
& \mathrm{subject\;to\;} & W\in \mathcal{W}_0^{x,y}.
\end{array}
\end{eqnarray*}
\end{problem}

Define the value function of Problem \ref{pmeanvariance} as follows:
\begin{align*}
V_1(x,y;z)\defi\inf_{W\in \mathcal{W}_0^{x,y}\atop
\BE[W]=z}\BE[W^2],\quad z\in \mathcal{D}.
\end{align*}
The following result, showing the connection between Problem \ref{pmeanvariance}
and Problem \ref{pmeanvariancewithoutconstrain}, can be proved by a standard
convex analysis argument.

\begin{prop}
\label{LagrangeMultiplier} Problem \ref{pmeanvariance} and Problem
\ref{pmeanvariancewithoutconstrain} have the following relations.

\begin{enumerate}
\item[(1)]  If $W_z^{*}$ solves Problem \ref{pmeanvariance} with parameter $z\in
\mathcal{D}$, then there exits $\ell \in \mathbb{R}$ such that $W_z^{*}$
also solves Problem \ref{pmeanvariancewithoutconstrain} with parameter $\ell
$.

\item[(2)]  Conversely, if $W_\ell $ solves Problem
\ref{pmeanvariancewithoutconstrain} with parameter $\ell \in \mathbb{R}$, then
it must also solve Problem \ref{pmeanvariance} with parameter $z=%
\BE[W_{\ell}]$.
\end{enumerate}
\end{prop}

It is easy to see that $\mathcal{W}_0^{x,y}-\ell =\mathcal{W}_0^{x-\ell
e^{-rT},y} $. As a consequence, we  consider the following problem instead of
Problem \ref{pmeanvariancewithoutconstrain}:

\begin{problem}
\label{pmeanvariancemainproblem}
\begin{eqnarray*}
\begin{array}{lll}
& \mini & \BE[W^2] \\
& \mathrm{subject\;to\;} & W\in \mathcal{W}_0^{x-\ell e^{-rT},y}
\end{array}
\end{eqnarray*}
\end{problem}

To solve the above problem, we use dynamic programming. In doing so we need to
parameterize the initial time. Consider the dynamics (\ref{xyprocess})--(\ref{xyprocess2})
where the initial time $0$ is revised to some $s\in[0,T)$, and
define $\mathcal{W}_s^{x,y}$ as the counterpart of $\mathcal{W}_0^{x,y}$ where the
initial time is $s$ and initial bond--stock position is $(x,y)$.
We then define the value function of Problem \ref{pmeanvariancemainproblem} as
\begin{align}
V(t,x,y)\defi\inf_{W\in \mathcal{W}_t^{x,y}}\BE[W^2],\quad (t,x,y)\in
[0,T)\times \mathbb{R}^2.  \label{Valuefunction}
\end{align}
The following proposition establishes a link between Problem
\ref{pmeanvariancemainproblem} and Problem \ref{pmeanvariance}.

\begin{prop}
\label{determineLagrangeMultiplier} If $z\in \mathcal{D}$, then
\begin{align*}
\sup\limits_{\ell \in \mathbb{R}}(V(0,x-\ell e^{-rT},y)-(\ell
-z)^2)=V_1(x,y;z)-z^2.
\end{align*}
\end{prop}

\begin{proof}
Note that
\begin{equation*}%\label{determineell}
\begin{split}
&\sup\limits_{\ell\in\R}(V(0,x-\ell e^{-rT},y)-(\ell-z)^2)
=\sup\limits_{\ell\in\R}\inf\limits_{W\in\mathcal{W}_0^{x-\ell e^{-rT},y}}\BE[W^2-(\ell-z)^2]\\
=&\sup\limits_{\ell\in\R}\inf\limits_{W\in\mathcal{W}_0^{x,y}}\BE[(W-\ell)^2-(\ell-z)^2]
\leqslant \sup\limits_{\ell\in\R}\inf\limits_{W\in
\mathcal{W}_0^{x,y}\atop \BE[W]=z}
\BE[(W-\ell)^2-(\ell-z)^2]\\
=&\sup\limits_{\ell\in\R}\inf\limits_{W\in\mathcal{W}_0^{x,y}\atop
\BE[W]=z}(\BE[W^2]-z^2) =\inf\limits_{W\in \mathcal{W}_0^{x,y}\atop
\BE[W]=z} (\BE[W^2]-z^2) =V_1(x,y;z)-z^2.
\end{split}
\end{equation*}
Therefore
\begin{align*}
\sup\limits_{\ell\in\R}(V(0,x-\ell e^{-rT},y)-(\ell-z)^2)&\leqslant
V_1(x,y;z)-z^2.
\end{align*}
\par
Since $V_1$ is convex and $z$ is an interior point of
$\mathcal{D}$, by convex analysis, there exists $\ell^*\in\R$ such
that
\begin{align*}
V_1(x,y;z)-2\ell^*z &\leqslant V_1(x,y;\tilde{z})-2\ell^*
\tilde{z},\quad \forall\; \tilde{z}\in\mathcal{D}.
\end{align*}
For any $W\in\mathcal{W}_0^{x,y}$, by the definition of $V_1$, we
have
\begin{eqnarray*}
\BE[(W-\ell^*)^2-(\ell^*-z)^2]=\BE[W^2]-2\ell^*(\BE[W]-z)-z^2
&\geqslant& V_1(x,y;\BE[W])-2\ell^*(\BE[W]-z)-z^2\\ &\geqslant&
V_1(x,y;z)-z^2.
\end{eqnarray*}
It follows
\begin{equation*}
\begin{split}
&\sup\limits_{\ell\in\R}(V(0,x-\ell e^{-rT},y)-(\ell-z)^2)
=\sup\limits_{\ell\in\R}\inf\limits_{W\in\mathcal{W}_0^{x,y}}\BE[(W-\ell)^2-(\ell-z)^2]\\
\geqslant
&\inf\limits_{W\in\mathcal{W}_0^{x,y}}\BE[(W-\ell^*)^2-(\ell^*-z)^2]
\geqslant V_1(x,y;z)-z^2,
\end{split}
\end{equation*}
which yields the desired result.
\end{proof}

Therefore, we need only to study the value function $V(t,x,y)$, which we set out to do now.

\begin{lem}
\label{valuefunctionproperty} The value function $V$ defined in (\ref{Valuefunction})
has the following properties.
\begin{enumerate}
\item[(1)]  For any $t\in [0,T)$, $V(t,\cdot ,\cdot )$ is convex and continuous in $%
\mathbb{R}^2$.

\item[(2)] For any $t\in [0,T)$, $V(t,x,y)$ is nonincreasing in $x$ and $y$.

\item[(3)] For any $\rho >0$, $t\in [0,T)$, we have $V(t,x+(1-\mu )\rho ,y-\rho
)\geqslant V(t,x,y)$, $V(t,x-(1+\lambda )\rho ,y+\rho )\geqslant V(t,x,y).$

\item[(4)] For any $\rho >0$, $t\in [0,T)$, we have $V(t,\rho x,\rho y)=\rho^2V(t,x,y)$.

\item[(5)] If $x+(1-\mu )y^{+}-(1+\lambda )y^{-}\geqslant 0$, then $V(t,x,y)=0$.
\end{enumerate}
\end{lem}

\begin{proof}
All the results can be easily proved in term of the definition of
$V$ and Lemma \ref{Wconvexset}.
\end{proof}

Due to part (5) of the above lemma, we only need to consider Problem
\ref{pmeanvariancemainproblem} in the \textit{insolvency region}
\begin{align*}
\mathscr{S}\defi\left\{ (x,y)\in \mathbb{R}^2\big|\;x+(1-\mu
)y^{+}-(1+\lambda )y^{-}<0\right\} .
\end{align*}
It is well known that the value function $V$ is a viscosity solution to the
following Hamilton-Jacobi-Bellman (HJB) equation or variational inequality (VI)
with terminal condition:
\begin{equation}\label{HJBP}
\begin{cases} \min\;\left\{\varphi_t+\mathcal{L}_0\varphi,
(1-\mu)\varphi_x-\varphi_y,\varphi_y-(1+\lambda)\varphi_x\right\}=0,
&\forall\;(t,x,y)\in[0,T)\times \mathscr{S},\\
\varphi(T,x,y)=(x+(1-\mu)y^+-(1+\lambda)y^-)^2, &
\forall\;(x,y)\in\mathscr{S},\\ \end{cases}  %\nonumber
\end{equation}
where
\begin{align*}
\mathcal{L}_0\varphi \defi \dfrac 12\sigma ^2y^2\varphi _{yy}+\alpha
y\varphi _y+rx\varphi _x.
\end{align*}

The idea of the subsequent analysis is to construct  a particular solution to the HJB equation, and then
employ the verification theorem to obtain an optimal strategy.
The construction of the solution is built upon a series of transformations on
equation (\ref{HJBP}) until we reach an equation related to the
so-called {\it double-obstacle} problem in physics which has been well studied in the partial
differential
equation literature.

We will show in Proposition \ref{vexistence} below that the constructed
solution $\varphi$ satisfies $\varphi _y-\left( 1+\lambda \right)\varphi _x=0$ when $y<0$.
Hence, we only need to focus on $y>0.$
A substantial technical difficulty arises  with the HJB equation (\ref{HJBP}) in that the
spatial variable $(x,y)$ is two dimensional. However,
the homogeneity of Lemma \ref{valuefunctionproperty}-(4)
 motivates us to make the transformation
\[
\varphi (t,x,y)=y^{{2}}\overline{V}(t,\frac xy),\text{ for }y>0,
\]
so as to reduce the dimension by one.
Accordingly, (\ref{HJBP}) is turned to
\[
\left\{
\begin{array}{ll}
\min \left\{ \overline{V}_t+\mathcal{L}_1\overline{V},(x+1-{\mu })\overline{V%
}_x-{2}\overline{V},-(x+1+\lambda )\overline{V}_x+{2}\overline{V}\right\}
=0,\ \ \  & \forall \ \left( t,x\right) \in [0,T)\times \mathscr{X}, \\
\overline{V}(T,x)=(x+1-{\mu })^{{2}}, & \forall \text{ }x\in %
\mathscr{X},
\end{array}
\right.
\]
where $\mathscr{X}\defi(-\infty,-(1-\mu)),$ and
\[
\mathcal{L}_1\overline{V}=\frac 12\sigma ^2x^2\overline{V}_{xx}-\left(
\alpha -r+\sigma ^2\right) x\overline{V}_x+\left( 2\alpha +\sigma ^2\right)
\overline{V}.
\]
Further, let
\begin{align*}
w(t,x)\defi\frac 12\ln \overline{V}(t,x),\quad (t,x)\in [0,T)\times %
\mathscr{X}.
\end{align*}
It is not hard to show that $w\left( t,x\right) $ is governed by
\begin{equation}\label{HJBeqn1}
\begin{cases} \min\;\left\{w_t+\mathcal{L}_2w, \frac{1}{w_x}-(x+1-\mu),
(x+1+\lambda)-\frac{1}{w_x} \right\}=0, &\forall\;(t,x)\in [0,T)\times
\mathscr{X},\\ w(t,x)=\ln(-x-(1-\mu)), & \forall\;x\in\mathscr{X},\\
\end{cases}
\end{equation}
where
\begin{align*}
\mathcal{L}_2w\defi \dfrac 12\sigma ^2x^2(w_{xx}+2w_x^2)-(\alpha -r+\sigma
^2)xw_x+\alpha +\frac 12\sigma ^2.
\end{align*}

\subsection{A Related Double-Obstacle Problem}

Equation (\ref{HJBeqn1}) is a variational inequality with gradient
constraints, which is hard to study. As in Dai and Yi (2009), we
will relate it to a double obstacle problem that is tractable. We
refer interested readers to Friedman (1988) for obstacle problems.

Let
\begin{equation}\label{vandphi}
v(t,x)\defi\frac 1{w_x(t,x)},\quad (t,x)\in [0,T)\times \mathscr{X}.
\end{equation}
Notice that
\[
\frac \partial {\partial x}\mathcal{L}_2w=-\frac 1{v^2}\left[ \frac 12\sigma
^2x^2v_{xx}-\left( \alpha -r\right) xv_x+\left( \alpha -r+\sigma ^2\right)
v+\sigma ^2\left( \frac{2x^2v_x-x^2v_x^2}v-2x\right) \right] .
\]
This {\it inspires} us to consider the following double-obstacle problem:
%\begin{problem}[\textbf{Double-Obstacle Problem}]
\begin{equation}\label{obstacleproblem}
\left\{
\begin{array}{l}
\max \left\{ \min \left\{ -v_t-\mathcal{L}v,v-\left( x+1-\mu \right)
\right\} ,v-\left( x+1+\lambda \right) \right\} =0,\hspace{0.2in}\forall
\;(t,x)\in [0,T)\times \mathscr{X} \\
v(T,x)=x+1-\mu ,\hspace{0.2in}\forall \;x\in \mathscr{X},
\end{array}
\right.
\end{equation}
where
\begin{equation}
\begin{split}
\mathcal{L}v\defi \dfrac 12\sigma ^2x^2v_{xx}-(\alpha -r)xv_x+(\alpha
-r+\sigma ^2)v+\sigma ^2\left( \frac{2x^2v_x-x^2v_x^2}v-2x\right) .
\end{split}
\label{calLdefinition}
\end{equation}

It should be emphasized that at this stage we have yet to know if
equation (\ref{obstacleproblem}) is mathematically equivalent to
equation (\ref{HJBeqn1}) via the transformation (\ref{vandphi}).
However, the following propositions show that
(\ref{obstacleproblem}) is solvable, and the solution to
(\ref{HJBeqn1}) can be constructed through the solutions of
(\ref{obstacleproblem}).

\begin{prop}
\label{vexistence} Equation (\ref{obstacleproblem}) has a solution
$v\in W_p^{1,2}([0,T)\times \left( -N,-\left( 1-\mu \right) \right)
),$ for any $N>-\left( 1-\mu \right) ,$ $p\in \left( 1,\infty
\right) $. Moreover,
\begin{equation}
v_t\leqslant 0,  \label{mono}
\end{equation}
\begin{equation}
0\leqslant v_x\leqslant 1,  \label{convex}
\end{equation}
and there exist two decreasing functions $x_s^{*}(\cdot)\in
C^\infty [0,T)$ and $x_b^{*}(\cdot)\in C^\infty [0,T_0)$ such that
\begin{equation}
\left\{ (t,x)\in [0,T)\times \mathscr{X}:\;v(t,x)=x+1-\mu \right\} =\left\{
(t,x)\in [0,T)\times \mathscr{X}:\;x\geqslant x_s^{*}(t)\right\}
\label{sell}
\end{equation}
and
\begin{equation}
\left\{ (t,x)\in [0,T)\times \mathscr{X}:\;v(t,x)=x+1+\lambda \right\}
=\left\{ (t,x)\in [0,T_0)\times \mathscr{X}:\;x\leqslant x_b^{*}(t)\right\}
\label{buy}
\end{equation}
where
\begin{align}
T_0=\max \left\{ T-\frac 1{\alpha -r}\ln \left( \frac{1+\lambda }{1-\mu }%
\right) ,0\right\} .  \label{deft0}
\end{align}
Further, we have
\begin{align}
\lim\limits_{t\uparrow T}x_s^{*}(t)=(1-\mu )x_M,\quad
\lim\limits_{T\uparrow \infty }x_s^{*}(0)=x_{s,\infty }^{*},\quad
\lim\limits_{t\uparrow T_0}x_b^{*}(t)=-\infty ,\quad
\lim\limits_{T\uparrow \infty }x_b^{*}(0)=x_{b,\infty }^{*},
\label{sameasdy}
\end{align}
where $x_M=-\frac{\alpha-r+\sigma^2}{\alpha-r}$, and $x_{s,\infty }^{*}$ and
$x_{b,\infty }^{*}$ are defined in (\ref{x_sdefinition}) and (\ref{x_bdefinition}).
\end{prop}

\begin{prop}
\label{vexistence2}
Define
\begin{equation}
\begin{split}
w(t,x)\defi\mathfrak{A}(t)+\ln (-x_s^{*}(t)-(1-\mu
))+\int_{x_s^{*}(t)}^x\!\frac 1{v(t,y)}\dd y,
\end{split}
\label{wdefinition}
\end{equation}
where
\begin{align}
\mathfrak{A}(t)\defi\int_t^T\!\dfrac{rx_s^{*2}(\tau )+(\alpha +r)(1-\mu
)x_s^{*}(\tau )+(\alpha +\tfrac 12\sigma ^2)(1-\mu )^2}{(x_s^{*}(\tau
)+1-\mu )^2}\dd\tau .  \label{Adefinition}
\end{align}
Then $w\in C^{1,2}\left( [0,T)\times \mathscr{X}\right) $ is a solution to
equation (\ref{HJBeqn1}).
Moreover,
for any $(t,x,y)\in [0,T)\times \mathscr{S}$, we define
\begin{equation}
\varphi (t,x,y)\defi \begin{cases}
y^2e^{2w(t,\tfrac{x}{y})}, &\textrm{if } y>0,\\
e^{2\mathfrak{B}(t)}(x+(1+\lambda)y)^2, &\textrm{if } y\leqslant 0
\end{cases},  \label{varphidefinition}
\end{equation}
where $w\left(t, x\right) $ is given in (\ref{wdefinition}) and
\begin{align*}
\mathfrak{B}(t)\defi\int_t^T\!\dfrac{rx_b^{*2}(\tau )+(\alpha +r)(1+\lambda
)x_b^{*}(\tau )+(\alpha +\tfrac 12\sigma ^2)(1+\lambda )^2}{(x_b^{*}(\tau
)+1+\lambda )^2}\dd\tau .
\end{align*}
Then $\varphi \in C^{1,2,2}([0,T)\times \mathscr{S}\setminus \{y=0\})$ is a
solution to the HJB equation (\ref{HJBP}).
\end{prop}

Due to their considerable technicality, the proofs of the preceding
two propositions are placed in Appendix A. Most of the above results
are similar to those obtained by Dai and Yi (2009) where they
considered the expected utility portfolio selection with transaction
costs. Nonetheless, there is one breakthrough made by the present
paper: both $x_s^{*}(\cdot)$ and $x_b^{*}(\cdot)$ are proven to be
$C^\infty $, whereas Dai and Yi (2009) only obtained the smoothness
of $x_s^{*}(\cdot).$ In fact, Dai and Yi (2009) essentially
considered a double obstacle problem for $w_x.$ In contrast, the
present paper takes the double obstacle problem for $1/w_x$ into
consideration. This seemly innocent modification in fact simplifies
the proof greatly. More importantly, it allows us to provide a
unified framework to obtain the smoothness of $x_s^{*}(\cdot)$ and
$x_b^{*}(\cdot).$ Later we will see that the smoothness of
$x_s^{*}(\cdot)$ and $x_b^{*}(\cdot)$ plays a critical role in the
proof of the existence of an optimal strategy.

Note that (\ref{convex}) is important in the study of the double obstacle problem
(\ref{obstacleproblem}). In essence, the result is based on parts (1) and (2) of
Lemma \ref{valuefunctionproperty}.

%Finally, since the terminal value of $\varphi$ defined in (\ref{varphidefinition})
%grows quadratically at infinity, it is not
%an easy task to show the uniqueness of solutions to (\ref{HJBP}) (in
%the sense of viscosity solution).
In the subsequent section, we plan to show that $\varphi (t,x,y)$
is nothing but the value function through the verification theorem
and a Skorokhod problem. At the same time, an optimal strategy
will be constructed.

\section{Skorokhod Problem and Optimal Strategy}

Due to (\ref{sell})--(\ref{buy}) and Proposition \ref{vexistence2}, we define
\begin{align}
\mathcal{SR}& =\left\{ (t,x,y)\in [0,T)\times \mathscr{S}\;\left| \;y>0,%
\hspace{1ex}x\geqslant x_s^{*}(t)y\right\} \right. ,  \nonumber \\
\mathcal{BR}& =\left\{ (t,x,y)\in [0,T)\times \mathscr{S}\;\left| \;y>0,%
\hspace{1ex}x\leqslant x_b^{*}(t)y,\text{ or }y\leqslant 0\right\} \right. ,\\
\mathcal{NT}& =\left\{ (t,x,y)\in [0,T)\times \mathscr{S}\;\left|\;y>0,
\;x_b^{*}(t)y<x<x_s^{*}(t)y\right\} ,\right.  \nonumber
\end{align}
which stand for the sell region, buy region and no trade region,
respectively. Here we set $x_b^{*}(t)=-\infty$ when $t\in[T_0,T]$
in view of the fact that $\lim\limits_{t\uparrow T_0}x_b^{*}(t)=-\infty$; hence
$\mathcal{BR}=\emptyset$ when $t\in [T_0,T]$.
 Notice that these regions do not depend on the target $z$.

\subsection{Skorokhod Problem and Verification Theorem}

In order to find the optimal solution to the MV problem, we need
to study the so-called Skorokhod problem.

\begin{problem}[\textbf{Skorokhod Problem}]
\label{skorohodproblem} Given $(0,X(0),Y(0))\in
\overline{\mathcal{NT}}$, find an admissible strategy $(M,N)$ such
that the corresponding bond--stock value process $(X,Y)$ is continuous in $[0,T]$,
and $(t,X(t),Y(t))\in \overline{\mathcal{NT}}$, for any $t\in [0,T]$.
\end{problem}

In other words, a solution to the Skorokhod problem is an investment
strategy with which the trading only takes place on the boundary of the no trade region.
It turns out the solution can be constructed via solving the following more specific problem.

\begin{problem}
\label{equivalentskorohod} Given $(0,X(0),Y(0))\in \overline{\mathcal{NT}}$,
find a process of bounded total variation $k$ and a continuous process $(X,Y)$ such that for any $t\in [0,T]$,
\begin{align*}
(t,X(t),Y(t))\in \overline{\mathcal{NT}}, \\
\dd X(t)=rX(t)\dd t+\gamma _1(X(t),Y(t))\dd |k|(t), \\
\dd Y(t)=\alpha Y(t)\dd t+\sigma Y(t)\dd B(t)+\gamma _2(X(t),Y(t))\dd |k|(t),
\\
|k|(t)=\int_0^t\!\!\mathbf{1}_{\{(s,X(s),Y(s))\in \partial {\mathcal{NT}}\}}%
\dd |k|(s),
\end{align*}
where $|k|(t)$ stands for the total variation of $%
k$ on $[0,t]$,
\[
(\gamma _1(x,y),\gamma _2(x,y))\defi
\begin{cases}
\tfrac{1}{\sqrt{(1+\lambda)^2+1}}(-(1+\lambda),1), &\textrm{if
}(t,x,y)\in\partial_1{\mathcal{NT}},\\
\tfrac{1}{\sqrt{(1-\mu)^2+1}}(1-\mu,-1), &\textrm{if
}(t,x,y)\in\partial_2{\mathcal{NT}}, \end{cases}
\]
and
\begin{align*}
\partial _1{\mathcal{NT}}\defi\left\{ (t,x,y)\in [0,T)\times \mathscr{S}%
\;\left| \;y>0,\hspace{1ex}x=x_b^{*}(t)y\right\} .\right.  \\
\partial _2{\mathcal{NT}}\defi\left\{ (t,x,y)\in [0,T)\times \mathscr{S}%
\;\left| \;y>0,\hspace{1ex}x=x_s^{*}(t)y\right\} .\right.
\end{align*}
\end{problem}

Letting a triplet $(k,X,Y)$ solves Problem \ref{equivalentskorohod},
define
\begin{align}
M(t)& \defi\tfrac 1{\sqrt{(1+\lambda )^2+1}}\int_0^t\!\!\mathbf{1}%
_{\{(s,X(s),Y(s))\in \partial _1\mathcal{NT}\}}\dd |k|(s),
\label{Mdefinition} \\
N(t)& \defi\tfrac 1{\sqrt{(1-\mu )^2+1}}\int_0^t\!\!\mathbf{1}%
_{\{(s,Y(s),Y(s))\in \partial _2\mathcal{NT}\}}\dd |k|(s).
\label{Ndefinition}
\end{align}
Then $(M,N)$ is a solution to the Skorokhod problem. Moreover, since $(M,N,X,Y)$
satisfies equations (\ref{xyprocess}) and (\ref{xyprocess2}), we can prove
that the corresponding terminal net wealth $W^{X,Y}(T)$ is the optimal solution to Problem
\ref{pmeanvariancemainproblem}. To prove that we need the verification theorem.

Note that one can extend naturally the definition of the Skorokhod problem to the time
horizon $[s,T]$, for any $s\in [0,T)$.

\begin{thm}[\textbf{Verification Theorem}]
\label{varificationtheorem}Let $\varphi $ be defined in (\ref{varphidefinition})
and $V$ be the value function defined in (\ref{Valuefunction}).
If the Skorokhod problem admits a solution in $[s,T]$, where $s\in [0,T)$, then
\begin{align*}
V(t,x,y)=\varphi (t,x,y),\;\forall \;(t,x,y)\in [s,T]\times \mathbb{R}^2.
\end{align*}
\end{thm}

\begin{proof} The proof is rather standard in the singular control literature.
We only give a sketch and refer interested
readers to Karatzas and Shreve (1998). Similar to Karatzas and Shreve (1998),
we can show that the
function $\varphi \leqslant V$ in $\overline{\mathcal{NT}}$. Moreover, if $%
\tau $ is a stopping time valued in $[s,T]$, where $s\in [0,T)$, and $(X,Y)$
is a solution to Problem \ref{skorohodproblem} in $[s,\tau ]$, then
\begin{align*}
\varphi (s,X(s),Y(s))=\BE[\varphi(\tau,X(\tau),Y(\tau))].
\end{align*}
Particularly, if $\tau =T$, then $\varphi =V$ in $\overline{\mathcal{NT}}%
\cap {[s,T]\times \mathbb{R}^2}$, and
\begin{align}
\label{phi=VonNT} V(s,X(s),Y(s))=\BE\left[ \left( W^{X,Y}(T)\right)
^2\right] .
\end{align}
The verification theorem in the $\mathcal{NT}$ follows. By Lemma \ref{valuefunctionproperty},
we know that $V(t,\cdot ,\cdot )$ is convex in $\mathbb{R}^2$.
So we can define its \textit{subdifferential} as
\[
\begin{split}
\partial V(t,x,y)\defi\left\{ (\delta _x,\delta _y)\big| V(t,\bar x,\bar
y)\geqslant V(t,x,y)+\delta _x\cdot (\bar x-x)\right. \left. +\delta _y\cdot
(\bar y-y),\forall(\bar x,\bar y)\in \mathbb{R}^2\right\} .
\end{split}
\]
Then, we are able to utilize the convex analysis as in Shreve and Soner
(1994) to obtain the verification theorem in $\mathcal{BR}$ and $\mathcal{SR}$.
\end{proof}

\subsection{Solution to Skorokhod Problem}

Note that, in the Skorokhod problem, Problem $\ref{skorohodproblem}$, the
reflection boundary depends on time $t$. This is very different from the standard Skorokhod problem
in the literature; see, e.g., Lions and Sznitman (1984). To remove the dependence of reflection boundary on time, we introduce a new state variable $Z(t)$ and instead consider an equivalent problem.

\begin{problem}
\label{equivalentskorohod2} Given $(0,X(0),Y(0))\in \overline{\mathcal{NT}}$,
find a process of bounded total variation $k$ and a continuous process $(Z,X,Y)$ such that, for any
$t\in [0,T]$,
\begin{align*}
(Z(t),X(t),Y(t))\in \overline{\mathcal{NT}}, \\
\dd X(t)=rX(t)\dd t+\gamma _1(X(t),Y(t))\dd |k|(t), \\
\dd Y(t)=\alpha Y(t)\dd t+\sigma Y(t)\dd B(t)+\gamma _2(X(t),Y(t))\dd |k|(t),
\\
\dd Z(t)=\dd t+\gamma _3(Z(t),X(t),Y(t))\dd |k|(t), \\
|k|(t)=\int_0^t\!\!\mathbf{1}_{\{(Z(s),X(s),Y(s))\in \partial {\mathcal{NT}}\}}%
\dd |k|(s),
\end{align*}
where $Z(0)=0$, $\gamma _3\equiv 0$.
\end{problem}

Clearly $Z(t)\equiv t$ because $\gamma_3(z,x,y)\equiv0$. Therefore
if $(k,Z,X,Y)$ solves Problem \ref{equivalentskorohod2}, then $(k,X,Y)$ solves Problem
\ref{equivalentskorohod}. It is worthwhile pointing out that the reflection boundary of Problem \ref{equivalentskorohod2} becomes time independent.

Now let us consider Problem \ref{equivalentskorohod2}.

\begin{thm}
\label{existenceofSkorokhodproblem} %Let $T_0$ be given in (\ref{deft0}).
There exists a unique solution to Problem \ref{equivalentskorohod2} in $[0,T]
$.
\end{thm}

\begin{proof} Both Lions and Sznitman (1984) and Dupus and Ishii (1993) have established existence and uniqueness for the Skorohod problem on a domain with sufficiently smooth boundary. Since the $C^{\infty}$ smoothness of $x^*_s(\cdot)$ and $x^*_b(\cdot)$ is in place, the proof is similar to that of Lemma 9.3 of Shreve and Soner (1994). It is worthwhile pointing out that Shreve and Soner (1994) did not concern the smoothness of $x^*_s$ and $x^*_b$ because they took
into consideration a stationary problem which leads to time-independent policies (free boundaries).
%We first consider the problem on $[0,T_0)$, and then
%consider it on $[T_0,T]$.
\end{proof}

Thanks to Theorem \ref{varificationtheorem} and Theorem
\ref{existenceofSkorokhodproblem}, we have

\begin{coro}
$V(t,x,y)=\varphi (t,x,y),\;\forall \;(t,x,y)\in [0,T]\times \mathbb{R}^2.$
\end{coro}

\section{Main Results}

\begin{thm}
\label{replicate} For any initial position $(x_0,y_0)\in \mathscr{S}$,
define
\[
(X(0),Y(0))\defi \begin{cases}
\left(\tfrac{(x_0+(1-\mu)y_0)x^*_s(0)}{x^*_s(0)+1-\mu},\tfrac{x_0+(1-%
\mu)y_0}{x^*_s(0)+1-\mu}\right), &\textrm{ if }(0,x_0,y_0)\in\mathcal{SR},\\
(x_0,y_0),&\textrm{ if }(0,x_0,y_0)\in\overline{\mathcal{NT}},\\
\left(\tfrac{(x_0+(1+\lambda)y_0)x^*_b(0)}{x^*_b(0)+1+\lambda},
\tfrac{x_0+(1+\lambda)y_0}{x^*_b(0)+1+\lambda}\right), &\textrm{ if
}(0,x_0,y_0)\in\mathcal{BR}. \end{cases}
\]
Let $(k,Z,X,Y)$ be the solution to Problem \ref{equivalentskorohod2} as stipulated
in Theorem \ref{existenceofSkorokhodproblem}. Then $%
(X,Y)$ is the unique solution to the Skorokhod problem, Problem
\ref{skorohodproblem}, and
\begin{align*}
V(0,x_0,y_0)=\BE\left[ \left( W^{X,Y}(T)\right) ^2\right] .
\end{align*}
Moreover, the optimal strategy $(M,N)$ is defined by (\ref{Mdefinition}) and
(\ref{Ndefinition}).
\end{thm}

\begin{proof}
Noting that $(0,X(0),Y(0))\in \overline{\mathcal{NT}}$, by
(\ref{phi=VonNT}),
\begin{align*}
V(0,X(0),Y(0))=\BE\left[\left(W^{X,Y}(T)\right)^2\right].
\end{align*}
Since $V=\varphi$, it is not hard to check that
\begin{align*}
V(0,x_0,y_0)=V(0,X(0),Y(0)).
\end{align*}
The proof is complete.
\end{proof}

As a final task before reaching the main result, we prove the existence of the Lagrange multiplier.

\begin{prop}
\label{existenceofLagrangeMultiplier} For any $(x,y,z)\in \mathbb{R}^2\times
\mathcal{D}$, there exists a unique $\ell ^{*}\in \mathbb{R}$ such that
\begin{align*}
V(0,x-\ell ^{*}e^{-rT},y)-(\ell ^{*}-z)^2=\sup\limits_{\ell \in \mathbb{R}%
}(V(0,x-\ell e^{-rT},y)-(\ell -z)^2).
\end{align*}
Moreover, $\ell ^{*}$ is determined by equation
\begin{align}
e^{-rT}V_x(0,x-\ell ^{*}e^{-rT},y)+2\ell ^{*}=2z.  \label{lequation}
\end{align}
\end{prop}

The proof is placed in Appendix B.
\begin{coro}\label{hatz=infty}
If $T>\frac 1{\alpha -r}\ln \left( \frac{1+\lambda }{1-\mu }\right)$, then
\begin{align*}
\hat{z}=+\infty,\quad\mathcal{D}=(e^{rT}x+(1-\mu)e^{rT}y^+-(1+\lambda)e^{rT}y^-,+\infty).
\end{align*}
\end{coro}
\begin{proof}
Suppose $\hat{z}<+\infty$. Then by the definition of $\hat{z}$,
for any $W\in\mathcal{W}_0^{x,y}$, we have $\BE[W]\leqslant \hat{z}$.
So for any $z>\hat{z}$, $\ell\geqslant 0$,
\begin{equation*}
\begin{split}
V(0,x-\ell e^{-rT},y)-(\ell-z)^2
&=\inf\limits_{W\in\mathcal{W}_0^{x-\ell e^{-rT},y}}\BE[W^2-(\ell-z)^2]
=\inf\limits_{W\in\mathcal{W}_0^{x,y}}\BE[(W-\ell)^2-(\ell-z)^2]\\
&=\inf\limits_{W\in\mathcal{W}_0^{x,y}}(\BE[W^2]+2\ell(z-\BE[W])-z^2)\\
&\geqslant 2\ell(z-\hat{z})-z^2.
\end{split}
\end{equation*}
Consequently,
\begin{align*}
\sup\limits_{\ell\in\R}(V(0,x-\ell e^{-rT},y)-(\ell-z)^2)=+\infty.
\end{align*}
However, the proof of Proposition \ref{existenceofLagrangeMultiplier} (Appendix B) shows that
the above supremum is finite under the condition
$T>\frac 1{\alpha -r}\ln \left( \frac{1+\lambda }{1-\mu }\right)$; see (\ref{vellstareq}).
The proof is complete.
\end{proof}

Now we arrive at the complete solution to the MV problem, Problem
\ref{pmeanvariance0}.

\begin{thm}
\label{mainresults} Problem \ref{pmeanvariance0}
admits an optimal solution if and only if $z\in \widetilde{\mathcal{D}}$,
where
\[
\widetilde{\mathcal{D}}\defi
\begin{cases} (e^{rT}x+(1-\mu)e^{rT}y^+-(1+\lambda)y^-,+\infty),
&\textrm{ if } T>\frac 1{\alpha -r}\ln \left( \frac{1+\lambda }{1-\mu }\right), \\
 (e^{rT}x+(1-\mu)e^{rT}y,e^{rT}x+(1-\mu)e^{\alpha T}y],
 &\textrm{ if } T\leqslant\frac 1{\alpha -r}\ln \left( \frac{1+\lambda }{1-\mu }\right) ,\; y>0,\\
\emptyset, &\textrm{ if } T\leqslant\frac 1{\alpha -r}\ln
\left( \frac{1+\lambda }{1-\mu }\right) ,\; y\leqslant 0. \end{cases}
\]
Moreover, if $z$ is on the boundary of $\widetilde{\mathcal{D}}$, i.e., $%
T\leqslant\frac 1{\alpha -r}\ln \left( \frac{1+\lambda }{1-\mu
}\right)$, $y>0$, and $z=e^{rT}x+(1-\mu )e^{\alpha T}y$, then the
optimal strategy is $(M,N)\equiv (0,0)$; otherwise, the value
function and the optimal solution are given by Theorem
\ref{replicate}, in which the initial position $(x_0,y_0)=(x-\ell
^{*}e^{-rT},y)$ where $\ell ^{*}$ is
determined by equation (\ref{lequation}).
\end{thm}

\begin{proof}
If $z\not\in \widetilde{\mathcal{D}}$,  then there is no feasible
solution by Theorem \ref{fesiblility}; so Problem \ref{pmeanvariance0} admits no optimal solution.
If $z=e^{rT}x+(1-\mu)e^{\alpha T}y$ while $T\leqslant\frac 1{\alpha
-r}\ln \left( \frac{1+\lambda }{1-\mu }\right)$ and $y>0$, then
Theorem \ref{fesiblility} again shows that the optimal strategy is $(M,N)\equiv(0,0)$.
\par
In all the other cases, it follows from Proposition \ref{existenceofLagrangeMultiplier} that
there exists a unique Lagrange multiplier $\ell^*$ such that
\begin{equation*}
\begin{split}
V(0,x-\ell^* e^{-rT},y)-(\ell^*-z)^2= \sup\limits_{\ell\in\R}
(V(0,x-\ell e^{-rT},y)-(\ell-z)^2).
\end{split}
\end{equation*}
Appealing to Proposition \ref{determineLagrangeMultiplier}, we have
\begin{equation*}
\begin{split}
V(0,x-\ell^* e^{-rT},y)-(\ell^*-z)^2=V_1(x,y;z)-z^2.
\end{split}
\end{equation*}
Theorem \ref{replicate} then dictates that there exists an admissible strategy
$(M^*,N^*)\in\mathcal{A}$ such that
\begin{align*}
V(0,x-\ell^* e^{-rT},y)=\BE\left[\left(W^{X^{x-\ell^*
e^{-rT},M^*,N^*},Y^{y,M^*,N^*}}(T)\right)^2\right].
\end{align*}
Noting that for any $(M,N)\in\mathcal{A}$, we have
\begin{align*}
X^{x-\ell^* e^{-rT},M,N}(T)&=X^{x,M,N}(T)-\ell^*,\\
Y^{y,M,N}(T)&=Y^{y,M,N}(T),\\
W^{X^{x-\ell^*e^{-rT},M,N},Y^{y,M,N}}(T)
&=W^{X^{x,M,N},Y^{y,M,N}}(T)-\ell^*;
\end{align*}
so
\begin{align*}
V(0,x-\ell^*
e^{-rT},y)=\BE\left[\left(W^{X^{x,M^*,N^*},Y^{y,M^*,N^*}}(T)-\ell^*\right)^2\right].
\end{align*}
By the definition of $V$, for any $(M,N)\in\mathcal{A}$,
\begin{align*}
V(0,x-\ell^* e^{-rT},y)
&\leqslant\BE\left[\left(W^{X^{x-\ell^* e^{-rT},M,N},Y^{y,M,N}}(T)\right)^2\right].\\
&=\BE\left[\left(W^{X^{x,M,N},Y^{y,M,N}}(T)-\ell^*\right)^2\right].
\end{align*}
Therefore $W^*$ is optimal to Problem
\ref{pmeanvariancewithoutconstrain}
 with parameter $\ell^*$, where $W^*$ is defined by
\begin{align*}
W^*=W^{X^{x,M^*,N^*},Y^{y,M^*,N^*}}(T).
\end{align*}
Owing to Proposition \ref{LagrangeMultiplier}, $W^*$ is optimal to
Problem \ref{pmeanvariance} with parameter $\BE[W^*]$. By the
uniqueness of $\ell^*$, we have $\BE[W^*]=z$. Thus $W^*$ is the
optimal solution to both Problems \ref{pmeanvariance} and
\ref{pmeanvariance0}, and $(M^*,N^*)$ is the optimal strategy.
\end{proof}

The preceding theorem fully describes the behavior of an optimal MV investor under
transaction costs. If the planning horizon is not long enough (the precise critical length depends
only on the stock excess return and the transaction fees), then what could be achieved at the
terminal time (in terms of the expected wealth) is rather limited. Otherwise,
any terminal target is achievable by an investment strategy, while an optimal (efficient)
strategy is to minimize the corresponding risk (represented by the variance). The optimal strategy
is characterized by three regions (those of sell, buy, and no trade) defined by (5.1).
The implementation of the strategy is very simple: a transaction takes
place only when the ``adjusted'' bond--stock process, $(X-\ell^* e^{-rT},Y)$,
reaches the boundary of the no-trade zone
so as for the process to stay within the zone (if initially the process is outside
of the no-trade zone then a transaction is carried out as in Theorem \ref{replicate} to move it
instantaneously into the no-trade zone).
%In other words, the investor tries {\it not} to trade unless absolutely necessary.
%This necessity, in turn, is completely specified by the sell and buy regions, defined in (5.1).
%This behavior is in sharp contrast with the case without transaction fees where an optimal
%mean--variance strategy is to trade {\it all the time}.

The optimal strategy presented here is markedly different from its no-transaction counterpart
[see, e.g., Zhou and Li (2000)]. With transaction costs, an investor tries {\it not}
to trade unless absolutely necessary, so as to keep the ``adjusted'' bond--stock ratio,
$\frac{x-\ell^* e^{-rT}}{y}$, between
the two barriers, $x^*_b(t)$ and $x^*_s(t)$, at any given time $t$. When there is no
transactions cost,
however, the two barriers coincide; so the optimal strategy is to keep the above ratio
exactly {\it at}
the barrier\footnote{In an EUM model--the Merton problem for example--the optimal solution is to keep
the bond--stock ratio exactly at a certain value. In the MV model the ratio must be ``adjusted'' in order
to account for the constraint of meeting the terminal target.}. This, in turn, requires the optimal strategy to trade {\it all the time}. Clearly,
the strategy presented here is more consistent with the actual investors' behaviors.

Let us examine more closely the trade zone consisting of the sell and buy regions, defined in (5.1).
By and large, when the adjusted bond--stock ratio, $\frac{x-\ell^* e^{-rT}}{y}$,
starts to be greater than a critical barrier (namely $x^*_s(t)$, which is time-varying), then
one needs to reduce the
stock holdings. When
the ratio starts to be smaller than {\it another} barrier ($x^*_b(t)$, again time-varying), then
one must accumulate the stock. It is interesting to see that $y\leqslant 0$ always triggers buying;
in other words shorting the stock is never favored, and any short position must be covered
immediately.
The essential reason behind this is the standing assumption that $\alpha>r$; so there is no
good reason to short the stock.
Nevertheless, this behavior is not seen in the absence of a transaction cost where
shorting is sometimes preferred even when $\alpha>r$.

Another not-so-obvious yet extremely intriguing behavior of the optimal strategy is that when
the time to maturity is short enough (precisely, when the remaining time is less than
$T^*=\frac 1{\alpha -r}\ln \left( \frac{%
1+\lambda }{1-\mu }\right)$), then one should not buy stock any
longer (unless to cover a possible short position). This is seen
from the fact that $\lim\limits_{t\uparrow T_0}x_b^{*}(t)=-\infty$
stated in Proposition \ref{vexistence} along with the definition of
the buy region. Moreover, since both the two barriers are decreasing
in time, the buy region gets smaller and the sell region gets bigger
as time passes along. This suggests that the investor would be less
likely to buy the stock and more likely to sell the stock when the
maturity date is getting closer. These phenomena, again, are in line
with what prevail in practice.

We end this section by a numerical example. Consider a market with the following
parameters:
\[
(\alpha ,\;r,\;\sigma ,\;\lambda ,\;\mu
,\;T)=(0.15,\;0.05,\;0.2,\;0.02,\;0.02,\;2).
\]
% An investor having the following initial position
% \[
% (x,\;y)=(-1,\;1)
% \]
% has an expected return at time $T$:
% \[
% z=1.1.
% \]

In this case, $T>\frac 1{\alpha -r}\ln \left( \frac{%
1+\lambda }{1-\mu }\right)$. In terms of a penalty method developed
by Dai and Zhong (2009), we numerically solve equation
(\ref{obstacleproblem}) and then construct the two free boundaries
by Proposition \ref{vexistence}.

\begin{center}
\includegraphics[width=8cm]{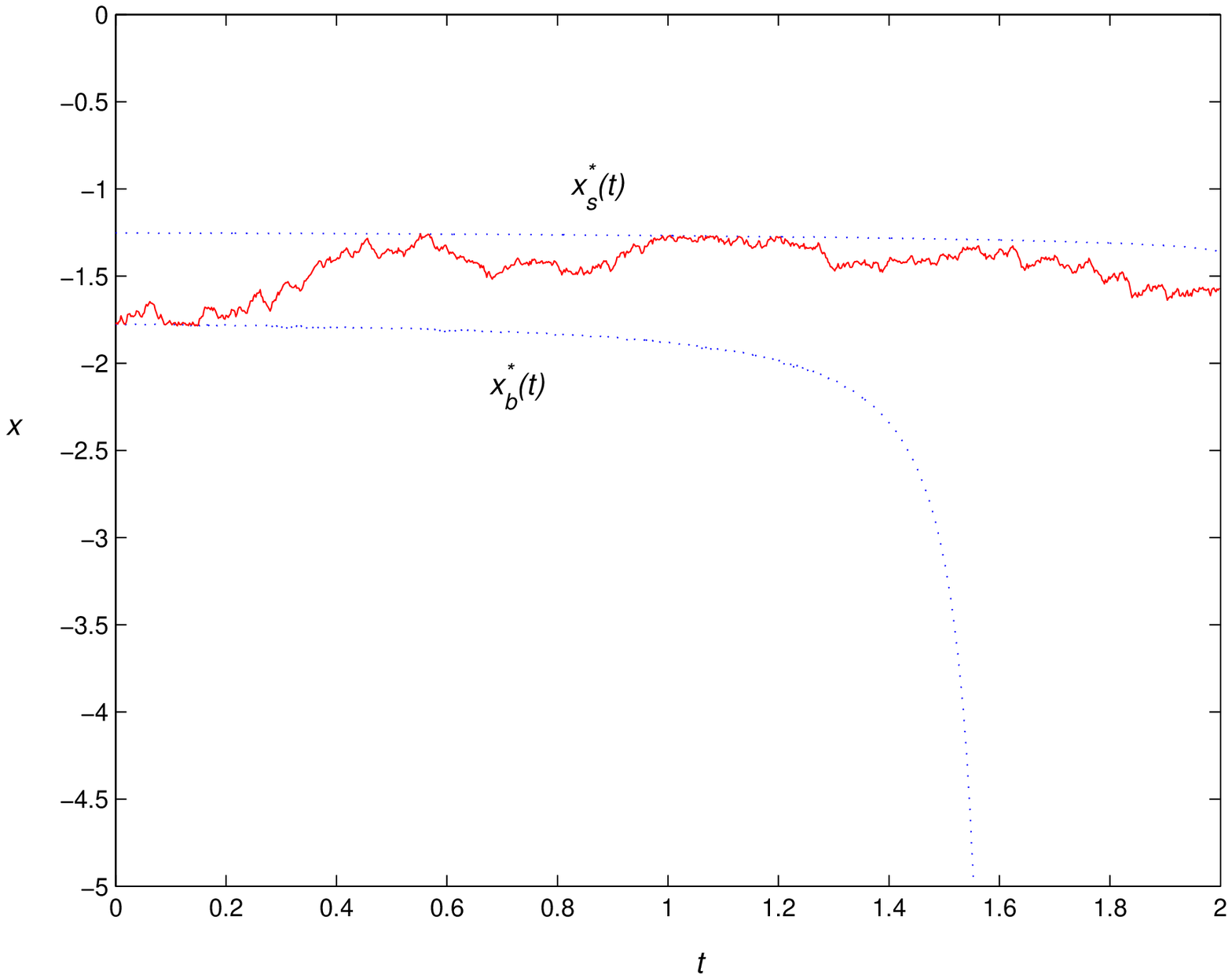}
\end{center}

Consider an investor having the  initial position $(x,\;y)=(-1,\;1)$
with an expected return $z=1.1$ at time $T$.
Based on (\ref{lequation}) we can calculate that $\ell^*=4.5069$; so the
adjusted initial position is $(x-\ell^*
e^{-rT},y)=(-5.078,1)$. The optimal strategy is the following. At time 0, applying
Theorem \ref{replicate}, the investor carries out a transaction so as to move his adjusted position to the boundary
of $\mathcal{NT}$. This is realized by buying 4.3395 units (in terms of the dollar amount) of the
stock, with the new adjusted position to be $(-9.5043,5.3395)$.
After the initial
time, the investor trades only on the boundaries of the $\mathcal{NT}$ region just to keep
his adjusted position within the $\mathcal{NT}$ region. Next consider another investor with an initial position $(x,\;y)=(1,0)$ and expected return $z=1.2$.
In this case  $\ell^*=2.3690$ and $(x-\ell^*
e^{-rT},y)=(-1.1436,0)$. So initially the investor buys 1.5047 worth of the stock moving the adjusted position to (-2.6784,1.5047) which is on the boundary of $\mathcal{NT}$. Afterwards, the trading strategy is simply to keep the adjusted position within $\mathcal{NT}$.

\section{Concluding Remarks}

This paper investigates a continuous-time Markowitz's mean--variance
portfolio selection model with proportional transaction costs. In
the terminology of stochastic control theory, this is a singular
control problem. We use the Lagrangian multiplier  and partial
differential equation to approach the problem.  The problem has been
completely solved in the following sense. First, the feasibility of
the model has been fully characterized by certain relationship among
the parameters. Second, the value function is given via a PDE, which
is analytically proven to be uniquely solvable and numerically
tractable, whereas the Lagrange multiplier is determined by an
algebra equation. Third, the optimal strategy is expressed in terms
of the free boundaries of the PDE. Economically, the results in the
paper have revealed three critical differences arising from the
presence of transaction costs. First,the expected return on the
portfolio may not be achievable if the time to maturity is not long
enough, while without transaction costs, any expected return can be
achieved in an arbitrarily short time. Second, instead of trading
all the time so as to keep a certain adjusted ratio between the
stock and bond to be a constant, there exists time-dependent upper
and lower boundaries so that transaction is only carried out when
the ratio is on the boundaries. Third, there is a critical time
which depends only on the stock excess return and the transaction
fees, such that beyond that time it is optimal not to buy stock at
all. Our results are closer to real investment practice where people
tend not to invest more in risky assets towards the end of the
investment horizon.

%\bigskip
%
%\noindent {\bf Acknowledgement.} The authors wish to thank the two anonymous referees for their detailed comments, especially for bringing up the issue of time inconsistency,
%that have led to a much improved version.

\appendix

\section{Proofs of Propositions \ref{vexistence} and \ref{vexistence2}}

Proposition \ref{vexistence2} is straightforward once Proposition
\ref{vexistence} is proved. So we prove Proposition \ref{vexistence}
only. Note that equation (\ref{obstacleproblem}) could have a
singularity if $v=0.$ To remove the possible singularity, let us
begin with the stationary counterpart of the problem. As in Theorem 6.1 of Dai and
Yi (2009), we are able to show that the semi-explicit stationary
solution is available through a Riccati equation
\[
\begin{cases} \mathcal{L}v_{\infty}=0, \textrm{ if }
x\in(x^*_{b,\infty},x^*_{s,\infty}),\\
v_{\infty}(x^*_{b,\infty})=x^*_{b,\infty}+1+\lambda,\quad
v'_{\infty}(x^*_{b,\infty})=1,\\
v_{\infty}(x^*_{s,\infty})=x^*_{s,\infty}+1-\mu,\quad
v'_{\infty}(x^*_{s,\infty})=1, \end{cases}
\]
with
\begin{align}
x_{s,\infty }^{*}& \defi-\frac a{a+k^{*}}(1-\mu ),
\label{x_sdefinition} \\
 x_{b,\infty }^{*}& \defi-\frac
a{a+\frac{k^{*}}{k^{*}-1}}(1+\lambda ), \label{x_bdefinition}
\end{align}
where $\mathcal{L}$ is defined in (\ref{calLdefinition}),
\begin{align*}
a\defi-\frac{2(\alpha -r+\sigma ^2)}{\sigma ^2}\in (-\infty ,-2),
\end{align*}
$k^{*}\in (1,2)$ is the solution to
\begin{align*}
F(k)=\frac{1+\lambda }{1-\mu },  %\label{Fequation}
\end{align*}
\[
F(k)\defi\begin{cases} \dfrac{a+\frac{k}{k-1}}{a+k}
\left(\dfrac{(c_1+\frac{k-1}{k}a)(c_2+\frac{1}{k}a)}
{(c_2+\frac{k-1}{k}a)(c_1+\frac{1}{k}a)}\right)^{\frac{1}{2(c_2-c_1)}},
&\textrm{if } \Delta_k<0,\\
\dfrac{a+\frac{k}{k-1}}{a+k}\exp\left(\dfrac{1}{2}
\left(\dfrac{1}{\frac{1}{k}a+\frac{1-a}{4}}-\dfrac{1}{\frac{k-1}{k}a+%
\frac{1-a}{4}}\right)\right), &\textrm{if } \Delta_k=0,\\
\dfrac{a+\frac{k}{k-1}}{a+k}\exp\left(\dfrac{1}{\sqrt{2\Delta_k}}
\left(\arctan\dfrac{k(a-1)-4a}{2k\sqrt{2\Delta_k}}
-\arctan\dfrac{4a-k(3a+1)}{2k\sqrt{2\Delta_k}}\right)\right),
&\textrm{if } \Delta_k>0, \end{cases}  \label{Fdefinition}
\]
$c_1$, $c_2$ are the two roots of $2c^2+(a-1)c+\frac{k-1}{k^2}a^2=0$, and
\begin{align*}
\Delta _k\defi\frac{k-1}{k^2}a^2-\frac 18(a-1)^2.
\end{align*}

Let $v(t,x)$ be a solution to equation (\ref{obstacleproblem}) restricted to the region $[0,T)\times
\left( -\infty ,x_{s,\infty }^{*}\right)$ with a boundary condition $v(t,x^*_{s,\infty})=x^*_{s,\infty}+1-\mu$.
Apparently, $v_\infty$ is a super-solution to equation (\ref{obstacleproblem}) in the region,
i.e., $v(t,x)\leq v_\infty(x)$ for all $x<-(1-\mu)$, $t\in [0,T)$.
It is easy to show that $v_\infty(x)$ is increasing in $x$. We then deduce
\begin{align}
v(t,x)\leq v_\infty(x_{s,\infty }^{*})=x_{s,\infty
}^{*}+1-\mu\defi -C_0<0 \ \ \text{for} \ \ x<x_{s,\infty
}^{*}.\label{eq1}
\end{align}
In what follows, we will confine equation (4.5) to the restricted region $[0,T)\times
\left( -\infty ,x_{s,\infty }^{*}\right)$ in which, due to (\ref{eq1}), the equation has no singularity. It is worthwhile pointing out that $v(t,x)$ can be trivially extended to the original region by letting $v(t,x)=x+1-\mu \ \ \text{for} \ \ x\geq x_{s,\infty }^{*}.$

In terms of a penalized approximation (see, for example, Friedman (1988)),
it is not hard to show that $v(t,x)\in
W_p^{1,2}\left( [0,T)\times (-N,x_{s,\infty }^{*})\right) $ for any $%
-N<x_{s,\infty }^{*},$ $p>1.$ By the maximum principle, (\ref{mono}) and (\ref
{convex}) follow. Then we have
\[
\frac \partial {\partial x}\left[ v-\left( x+1-\mu \right) \right] =\frac
\partial {\partial x}\left[ v-\left( x+1+\lambda \right) \right]
=v_x-1\leqslant 0,
\]
which implies the existence of $x_s(\cdot)$ (or $x_b(\cdot))$ as a
single-value function. The monotonicity of $x_s(\cdot)$ and
$x_b(\cdot)$ is due to
\[
\frac \partial {\partial t}\left[ v-\left( x+1-\mu \right) \right] =\frac
\partial {\partial t}\left[ v-\left( x+1+\lambda \right) \right]
=v_t\leqslant 0.
\]
The proof of ({\ref{sameasdy}) is the same as that of Theorem 4.5 and 4.7 of Dai and Yi (2009).

It remains to show the smoothness of $x_b^{*}(\cdot)$ and $%
x_s^{*}(\cdot). $ To begin with, let us make a transformation and introduce two
lemmas.

Let $z=\log \left( -x\right)$, $u(t,z)=v(t,x)$. Then
\[
\left\{
\begin{array}{l}
\max \left\{ \min \left\{ -u_t-\mathcal{L}_1u,u-\left( -e^z+1-\mu \right)
\right\} ,u-\left( -e^z+1+\lambda \right) \right\} =0, \\
u(T,z)=-e^z+1-\mu ,\hspace{1.0in}\left( t,z\right) \in [0,T)\times \mathscr{Z}, \\
u\left( t,\log \left( -x^{*}_{s,\infty}\right) \right)
=x^{*}_{s,\infty}+1-\mu,
\end{array}
\right.
\]
where $\mathscr{Z}=(\log \left( -x^{*}_{s,\infty}\right) ,+\infty ),$ and
\[
\mathcal{L}_1u=\frac{\sigma ^2}2u_{zz}-\left( \alpha -r+\frac{\sigma ^2}%
2\right) u_z+\left( \alpha -r+\sigma ^2\right) u\ -\sigma ^2\left[ \frac{%
u_z^2}u+2e^z\frac{u_z}u-2e^z\right] .
\]

\begin{lem}
For any $\left( t,z\right) \in [0,T)\times \mathscr{Z}$, we have

(1) $u\leqslant -C_0;$

(2) $u_z=xv_x\geqslant x=-e^z,$ i.e., $u_z+e^z\geqslant 0;$

(3) $u-u_z\geqslant 0.$ Moreover, there is a constant $C_1>0$ such that $%
u-u_z\geqslant C_1.$
\end{lem}

\noindent {\it Proof.} Part (1) and (2) are immediate from
(\ref{eq1}) and (\ref{convex}). Let us prove Part (3). Denote
$w=u_z.$ So,
\begin{eqnarray*}
\frac \partial {\partial z}\left( -u_t-\mathcal{L}_1u\right) =-w_t-\mathcal{L%
}_2w+2e^z\sigma ^2\left( \frac{u_z}u-1\right)
\end{eqnarray*}
where
\[
\mathcal{L}_2w=\frac{\sigma ^2}2w_{zz}-\left( \alpha -r+\frac{\sigma ^2}%
2\right) w_z+\left( \alpha -r+\sigma ^2\right) w-\sigma ^2\left( \frac{%
2\left( u_z+e^z\right) }uw_z-\frac{u_z\left( u_z+2e^z\right) }{u^2}w\right).
\]
We define
\begin{align*}
\text{SR}& =\left\{(t,z) \in [0,T)\times \mathscr{Z}\;\left|\;u=-e^z+1-\mu\right\}\right. ,  \nonumber \\
\text{BR}& =\left\{(t,z) \in [0,T)\times \mathscr{Z}\;\left|\;u=-e^z+1+\lambda\right\} \right. ,\\
\text{NT}& =\left\{(t,z) \in [0,T)\times
\mathscr{Z}\;\left|\;-e^z+1-\mu < u< -e^z+1+\lambda \right\}
.\right. \nonumber
\end{align*}

Notice that we can rewrite
\[
-u_t-\mathcal{L}_1u=-u_t-\mathcal{L}_2u+2e^z\sigma ^2\left( \frac{u_z}%
u-1\right) .
\]
Denote $H=u-u_z.$ Then
\begin{align*}
\ -H_t-\mathcal{L}_2H=0\text{ in NT}   \label{eq3}
\end{align*}
Clearly $H=1-\mu \text{ in SR}\text{ and at }t=T$, and
$H=1+\lambda \text{ in BR}$. Hence, applying the maximum principle
yields $H\geq 0$. Moreover, it is not hard to verify that the
coefficients in $\mathcal{L}_2H$ are bounded. We then infer that
there is a constant $C_1>0,$ such that $H\geqslant C_1.$

\begin{lem}
There is a constant $C_2>0,$ such that $u_t\geqslant -C_2.$
\end{lem}

\noindent {\it Proof.} Let $z_s(t)=\log(-x^*_s(t))$ be corresponding the selling boundary. For
$z>z_s(t)$
\begin{eqnarray*}
u_t|_{t=T} &=&-\mathcal{L}_1\left( -e^z+1-\mu \right) =-\left(
\alpha -r\right) \left(
1-\mu \right) -\frac{\left( 1-\mu \right) ^2}{-e^z+1-\mu } \\
&\geqslant &-\left( \alpha -r\right) (1-\mu ).
\end{eqnarray*}
Applying the maximum principle gives the desired result.

\hspace{1.0in}

We are now to prove that both $z_s(\cdot)$ and $z_b(\cdot)$ are
$C^\infty,$ where $z_s(t)=\log(-x^*_s(t))$ and $z_b(t)=\log(-x^*_b(t))$.
Thanks to the bootstrap technique, we only need to show that
they are Lipschitz-continuous. Hence, it suffices to prove the cone property, namely,
for any $\left( t,z_0\right) \in [0,T)\times \mathscr{Z}$, there exists a
constant $C>0$ such that
\begin{eqnarray*}
\left. (T-t)u_t+C\frac \partial {\partial z}\Big( u-\left( -e^z+1-\mu
\right) \Big)\right| _{(t,z_0)} &\geqslant &0, \\
\left. (T-t)u_t+C\frac \partial {\partial z}\Big( u-\left( -e^z+1+\lambda
\right) \Big)\right| _{(t,z_0)} &\geqslant &0,
\end{eqnarray*}
which is equivalent to
\begin{equation}
\left. (T-t)u_t+C\left( u_z+e^z\right) \right| _{(t,z_0)}\geqslant 0.  \label{eq4}
\end{equation}
Now let us prove (\ref{eq4}). We can only focus on the NT region. Note that
\[
\frac \partial {\partial t}\bigg( -u_t-\mathcal{L}_1u\bigg)
=\left( -\frac \partial {\partial t}-\mathcal{L}_2\right) u_t.
\]
It follows
\[
\left( -\frac \partial {\partial t}-\mathcal{L}_2\right) \left[ \left(
T-t\right) u_t\right] =u_t\text{ in NT.}
\]
On the other hand, it is not hard to check
\begin{eqnarray*}
\left( -\frac \partial {\partial t}-\mathcal{L}_2\right) \left(
u_z+e^z\right) &=&\sigma ^2e^z\frac{u-u_z}{u^2}\left( u+u_z+2e^z\right) \\
\ &\geqslant &\sigma ^2e^z\frac{u-u_z}{u^2}\left( u+u_z-2u_z\right) \\
\ &=&\sigma ^2e^z\frac{\left( u-u_z\right) ^2}{u^2}\geqslant \frac{C_1^2\sigma
^2e^z}{u^2},\text{ in NT}
\end{eqnarray*}
where $u-u_z\geqslant C_1$ is used in the last inequality. Thus,
\begin{eqnarray*}
&&\ \ \left( -\frac \partial {\partial t}-\mathcal{L}_2\right) \left[
(T-t)u_t+C\left( u_z+e^z\right) \right] \\
\ &\geqslant &u_t+C\frac{C_1^2\sigma ^2e^z}{u^2}\geqslant -C_2+C\frac{C_1^2\sigma ^2e^z%
}{u^2},\text{ in NT.}
\end{eqnarray*}
%In a bounded domain in which $u$ is bounded, we could choose $C$ big enough
%such that $-C_2+C\frac{C_1^2\sigma ^2e^z}{u^2}\geqslant 0.$
Since NT is unbounded, we can follow Soner and Shreve (1991) to
introduce an auxiliary function $\psi (t,z;z_0)=e^{a(T-t)}\left(
z-z_0\right) ^2$ with a constant $a>0.$ We can choose $a$ big
enough so that
\[
\left( -\frac \partial {\partial t}-\mathcal{L}_2\right) \psi \left(
t,z;z_0\right) \geqslant C_3\left( z-z_0\right) ^2-C_4,
\]
where $C_3$ and $C_4$ are positive constants independent of $(t,z)$. It
follows $\ \ $%
\begin{eqnarray*}
&&\left( -\frac \partial {\partial t}-\mathcal{L}_2\right) \left[
(T-t)u_t+C\left( u_z+e^z\right) +\psi \left( t,z;z_0\right) \right] \\
&\geqslant &-C_2+C\frac{C_1^2\sigma ^2e^z}{u^2}+C_3\left( z-z_0\right) ^2-C_4.
\end{eqnarray*}
Then we can choose $r>0$ such that
\[
C_3r^2-C_2-C_4\geqslant 0
\]
and choose $C>0$ big enough such that
\[
C\frac{C_1^2\sigma ^2e^z}{u^2}-C_2-C_4\geqslant 0\text{ for }\left| z-z_0\right|
\leqslant r.
\]
It then follows
\[
\left( -\frac \partial {\partial t}-\mathcal{L}_2\right) \left[
(T-t)u_t+C\left( u_z+e^z\right) +\psi \left( t,z;z_0\right) \right] \geqslant 0,%
\text{ in NT.}
\]
Applying the maximum principle and penalty approximation, we
conclude
\[
(T-t)u_t+C\left( u_z+e^z\right) +\psi \left( t,z;z_0\right) \geqslant 0,\text{ }%
(t,z)\in [0,T)\times \mathscr{Z}.
\]
Letting $z=z_0,$ we get the desired result.

\section{Proof of Proposition \ref{existenceofLagrangeMultiplier}}

\begin{proof}
From
\begin{equation*}
\begin{split}
V(0,x-\ell e^{-rT},y)-(\ell-z)^2
&=\inf\limits_{W\in\mathcal{W}_0^{x-\ell
e^{-rT},y}}\BE[W^2-(\ell-z)^2]=\inf\limits_{W\in\mathcal{W}_0^{x,y}}\BE[(W-\ell)^2-(\ell-z)^2]\\
&=\inf\limits_{W\in\mathcal{W}_0^{x,y}}(\BE[W^2]-2\ell(\BE[W]-z)),
\end{split}
\end{equation*}
it follows that $V(0,x-\ell e^{-rT},y)-(\ell-z)^2$ is concave in
$\ell$. So its maximum attains at point $\ell^*$ which satisfies
\begin{align*}
\frac{\partial}{\partial \ell}\Big( V(0,x-\ell
e^{-rT},y)-(\ell-z)^2\Big)\bigg|_{\ell=\ell^*}=0,
\end{align*}
i.e.,
\begin{align*}
e^{-rT}V_x(0,x-\ell^* e^{-rT},y)+2\ell^*=2z.
\end{align*}
Define
\begin{align*}
f(\ell)\defi e^{-rT}V_x(0,x-\ell e^{-rT},y)+2\ell.
\end{align*}
Then by the convexity of $V(0,x-\ell e^{-rT},y)-(\ell-z)^2$ in
$\ell$, we have that $f$ is increasing. Since $V_x\leqslant 0$, we
have $f(z)\leqslant 2z$. By the monotonicity of $f$, the existence
of $\ell^*$ depends on $\lim\limits_{\ell\to+\infty}f(\ell)$.
\begin{itemize}
\item We first consider the case when $T_0>0$. In this case
$x^*_b(0)\in(-\infty,0)$. If $y\leqslant 0$, then
\begin{align*}
(0,x-\ell e^{-rT} ,y)\in \mathcal{BR},\quad\forall\;\ell\geqslant
z.
\end{align*}
If $y>0$, then
\begin{align*}
(0,x-\ell e^{-rT} ,y)\in \mathcal{BR},\quad\forall\;\ell\geqslant
e^{rT}(x-x^*_b(0)y).
\end{align*}
Therefore,
\begin{equation*}
\begin{split}
\lim_{\ell\to+\infty}f(\ell)
&=\lim_{\ell\to+\infty}\left(e^{-rT}V_x\left(0,x-\ell e^{-rT},y\right)+2\ell\right)\\
&=\lim_{\ell\to+\infty}
\left(2e^{-rT}e^{2\mathfrak{B}(0)}\left(x-\ell e^{-rT}+(1+\lambda)y\right)+2\ell\right)\\
&=\lim_{\ell\to+\infty} 2\left(1-e^{-2rT+2\mathfrak{B}(0)}\right)\ell+2e^{-rT}(x+(1+\lambda)y)\\
&=+\infty,
\end{split}
\end{equation*}
where we have used the fact that $\mathfrak{B}(0)<rT$ when $T_0>0$.
Therefore, for any
\begin{align*}
z\in(e^{rT}x+(1-\mu)e^{rT}y^+-(1+\lambda)e^{rT}y^-,+\infty)
\end{align*}
there exists $\ell^*$ such that
\begin{align}
e^{-rT}V_x(0,x-\ell^* e^{-rT},y)+2\ell^*&=2z,\nonumber\\
V\left(0,x-\ell^* e^{-rT},y\right)-(\ell^*-z)^2&=
\sup\limits_{\ell\in\R}(V(0,x-\ell
e^{-rT},y)-(\ell-z)^2).\label{vellstareq}
\end{align}
Now we prove the uniqueness. For $T_0>0$, we have
\begin{align*}
\mathfrak{A}(0)<rT,\quad
\mathfrak{B}(0)<rT.
\end{align*}
If $(0,x-\ell e^{-rT},y)\in\mathcal{SR}$, then
\begin{align*}
f'(\ell)=-e^{-2rT}V_{xx}(0,x-\ell
e^{-rT},y)+2=-2e^{-2rT+2\mathfrak{A}(0)}+2>0.
\end{align*}
Similarly, if $(0,x-\ell e^{-rT},y)\in\mathcal{BR}$, then
\begin{align*}
f'(\ell)=-e^{-2rT}V_{xx}(0,x-\ell
e^{-rT},y)+2=-2e^{-2rT+2\mathfrak{B}(0)}+2>0.
\end{align*}
By the maximum principle, we have
\begin{align*}
  f'(\ell)>0,\quad \textrm{for }(0,x-\ell e^{-rT},y)\in\mathcal{NT}.
\end{align*}
This implies the uniqueness of $\ell^*$. \item Now, we move to the
case when $T_0=0$. According to Theorem \ref{fesiblility}, we have
\begin{equation*}
\mathcal{D}=
\begin{cases}
(e^{rT}x+(1-\mu)e^{rT}y,e^{rT}x+(1-\mu)e^{\alpha T}y),
&\textrm{ if } y>0,\\
\emptyset, &\textrm{ if }y\leqslant 0.
\end{cases}
\end{equation*}
We only need to consider the case of $y>0$. Note that in this
case,
\begin{align*}
(0,x-\ell e^{-rT},y)\in\mathcal{NT},\quad\forall\;\ell\geqslant
e^{rT}(x-x^*_s(0)y).
\end{align*}
By the homogeneity property, we have
\begin{align*}
  V_x(t,\rho x,\rho y)=\rho V_x(t,x,y),\quad \forall\;(t,x,y,\rho)\in[0,T)\times\R^2\times\R_+.
\end{align*}
So we can make the following transformation in $\mathcal{NT}$,
\begin{align*}
z=-\frac{y}{x}\in \left(0,\tfrac{-1}{x^*_s(0)}\right), \ \ \
\bar{v}(t,z)=-\frac{1}{x}V_x(t,x,y).
\end{align*}
Then
\begin{equation*}
\begin{cases}
\bar{v}_t+\frac{1}{2}\sigma^2z^2\bar{v}_{zz}+(\alpha-r)z\bar{v}_z+2r\bar{v}=0,\ \
(t,z)\in[0,T)\times \left(0,\tfrac{-1}{x^*_s(0)}\right).\\%[3ex]
\bar{v}(T,z)=2(-1+(1-\mu)z).
\end{cases}
\end{equation*}
Therefore
\begin{align*}
\bar{v}(t,0)=-2e^{2r(T-t)}.
\end{align*}
Let $\tilde{v}(t,z)=\bar{v}_z(t,z),$ which satisfies
\begin{equation*}
\begin{cases}
\tilde{v}_t+\frac{1}{2}\sigma^2z^2\tilde{v}_{zz}
+(\alpha-r+\sigma^2)z\tilde{v}_z+(\alpha+r)\tilde{v}=0,\ \
\hspace{1.5ex} (t,z)\in[0,T)\times
\left(0,\tfrac{-1}{x^*_s(0)}\right),\\
\tilde{v}(T,z)=2(1-\mu).
\end{cases}
\end{equation*}
Therefore
\begin{align*}
\tilde{v}(t,0)=2(1-\mu)e^{(\alpha+r)(T-t)}.
\end{align*}
It follows
\begin{equation*}
\begin{split}
V_x(0,x,y)&=-x\bar{v}\left(0,-\tfrac{y}{x}\right)
=-x\left(\bar{v}(0,0)-\tfrac{y}{x}\bar{v}_z(0,0)+O\left(\tfrac{y^2}{x^2}\right)\right)\\
&=2xe^{2rT}+2y(1-\mu)e^{(\alpha+r)T}+O\left(\tfrac{y^2}{|x|}\right).
\end{split}
\end{equation*}
So
\begin{equation*}
\begin{split}
\lim\limits_{\ell\to +\infty}f(\ell)
&=\lim\limits_{\ell\to +\infty}\left(e^{-rT}V_x(0,x-\ell e^{-rT},y)+2\ell\right)\\
&=\lim\limits_{\ell\to +\infty} \Bigg(e^{-rT}\left(2(x-\ell
e^{-rT})e^{2rT}+2y(1-\mu)e^{(\alpha+r)T}\right)
+O\left(\frac{y^2}{\left|x-\ell e^{-rT}\right|}\right)+2\ell\Bigg)\\
&=2(e^{rT}x+(1-\mu)e^{\alpha T}y).
\end{split}
\end{equation*}
The monotonicity of $f$ ensures the existence of $\ell^*$. The
proof for the uniqueness is similar as above.
\end{itemize}
The proof is complete.
\end{proof}

\end{document}